\newif\ifdoublecol

\doublecoltrue
\documentclass[conference,letterpaper]{IEEEtran}
\IEEEoverridecommandlockouts

\usepackage{etex} % Solves: "no room for new dimensions"
\ifCLASSINFOpdf
  \usepackage[pdftex]{graphicx}

  \usepackage{epstopdf}
\else
  \usepackage[dvips]{graphicx}
  \graphicspath{{../eps/}}
  \DeclareGraphicsExtensions{.eps}
\fi
\usepackage{cite}

\usepackage[cmex10]{amsmath}
\usepackage{xcolor,colortbl}
\definecolor{col1}{HTML}{3891A6}
\definecolor{col2}{HTML}{EF5B5B}
\definecolor{col3}{HTML}{3DDC97}

\usepackage{amsmath,amsthm,relsize}
\usepackage{amsfonts, amssymb, cuted}
\usepackage{mathtools}
\usepackage{algorithm}
\usepackage[noend]{algpseudocode}
\usepackage{cite}
\usepackage{soul}
\usepackage{graphicx}
\usepackage{tikz}
\usepackage{pgfplotstable}
\usepackage{pgfplots}
\usepackage{nicefrac}
\usepackage{enumitem}
\usepackage{support-caption}
\usepackage{subcaption}
\usepackage[font=footnotesize]{caption}
\usepackage{multirow}
\pgfplotsset{compat=1.15}
\usepackage{relsize}
\usetikzlibrary{patterns}
\usepackage{acro}
\usepackage{pifont}
\newtheorem{theorem}{Theorem}

\usepackage{todonotes}

\usepackage{mathtools}

\usepackage{array}
\usepackage{booktabs}

\usepackage{diagbox}
\usepackage{smartdiagram}
\usesmartdiagramlibrary{additions}
\usepackage{bm}
\usepackage{multicol}

\usepackage{tabularx}
\usepackage{colortbl}

\usepackage{hyperref}
\hypersetup{
    colorlinks=true,
    linkcolor=blue,
    citecolor=blue,
    filecolor=blue,
    urlcolor=blue
}

\usetikzlibrary{plotmarks}
  \usetikzlibrary{arrows.meta}
  \usepgfplotslibrary{patchplots}
  \usepackage{grffile}
  \pgfplotsset{plot coordinates/math parser=false}
  \newlength\figureheight
  \newlength\figurewidth
   \pgfplotsset{compat=1.11,
    /pgfplots/ybar legend/.style={
    /pgfplots/legend image code/.code={%
       \draw[##1,/tikz/.cd,yshift=-0.25em]
        (0cm,0cm) rectangle (3em,8pt);},
   },
}
\usepgfplotslibrary{groupplots,units}
\pgfplotsset{
  compat=1.9,
  unit code/.code 2 args={\si{#1#2}} % from manual, for using siunitx to typeset units
}
\usepackage{siunitx}

\usepackage{color}
\usepackage{colortbl}
\usepackage{multirow}

\newlength{\Oldarrayrulewidth}
\definecolor{intnull}{RGB}{213,229,255}
\definecolor{inteins}{RGB}{128,179,255}
\definecolor{intzwei}{RGB}{42,127,255}
\definecolor{intdrei}{RGB}{0,85,212}
\definecolor{intvier}{RGB}{0,51,128}
\definecolor{intfunf}{RGB}{0,34,85}
\usepackage{arydshln} % For dashed lines

\usetikzlibrary{decorations.pathreplacing}
\renewcommand*{\arraystretch}{.4}

\newtheorem{lemma}{Lemma}

\usepackage{arydshln}

\newcommand{\herm}{^{\mbox{\scriptsize H}}}

\newcommand{\vbar}{\raisebox{.17ex}{\rule{.04em}{1.35ex}}}
\newcommand{\vbarind}{\raisebox{.01ex}{\rule{.04em}{1.1ex}}}
\newcommand{\R}{\ifmmode{\rm I}\hspace{-.2em}{\rm R} \else ${\rm I}\hspace{-.2em}{\rm R}$ \fi}
\newcommand{\T}{\ifmmode{\rm I}\hspace{-.2em}{\rm T} \else ${\rm I}\hspace{-.2em}{\rm T}$ \fi}
\newcommand{\N}{\ifmmode{\rm I}\hspace{-.2em}{\rm N} \else \mbox{${\rm I}\hspace{-.2em}{\rm N}$} \fi}
\newcommand{\B}{\ifmmode{\rm I}\hspace{-.2em}{\rm B} \else \mbox{${\rm I}\hspace{-.2em}{\rm B}$} \fi}
\newcommand{\Hil}{\ifmmode{\rm I}\hspace{-.2em}{\rm H} \else \mbox{${\rm I}\hspace{-.2em}{\rm H}$} \fi}
\newcommand{\C}{\ifmmode\hspace{.2em}\vbar\hspace{-.31em}{\rm C} \else \mbox{$\hspace{.2em}\vbar\hspace{-.31em}{\rm C}$} \fi}
\newcommand{\Cind}{\ifmmode\hspace{.2em}\vbarind\hspace{-.25em}{\rm C} \else \mbox{$\hspace{.2em}\vbarind\hspace{-.25em}{\rm C}$} \fi}
\newcommand{\Q}{\ifmmode\hspace{.2em}\vbar\hspace{-.31em}{\rm Q} \else \mbox{$\hspace{.2em}\vbar\hspace{-.31em}{\rm Q}$} \fi}
\newcommand{\Z}{\ifmmode{\rm Z}\hspace{-.28em}{\rm Z} \else ${\rm Z}\hspace{-.28em}{\rm Z}$ \fi}

\newtheorem{exmp}{Example}%[section]
\theoremstyle{definition}

\newtheorem{rem}{Remark}

\newcommand{\CF}[0]{{\mathcal{F}}}

\newcommand{\CJ}[0]{{\mathcal{J}}}
\newcommand{\CK}[0]{{\mathcal{K}}}

\newcommand{\CM}[0]{{\mathcal{M}}}
\newcommand{\CN}[0]{{\mathcal{N}}}

\newcommand{\CP}[0]{{\mathcal{P}}}

\newcommand{\CT}[0]{{\mathcal{T}}}

\newcommand{\Bw}[0]{{\mathbf{w}}}
\newcommand{\Bx}[0]{{\mathbf{x}}}

\newcommand{\BH}[0]{{\mathbf{H}}}

\newcommand{\BU}[0]{{\mathbf{U}}}

\newcommand{\SfP}[0]{{\mathsf{P}}}

\usepackage{tikz}

\usepackage{tikz}
\usetikzlibrary{arrows,
                calc,
                decorations.pathreplacing,
                calligraphy,% had to be after decorations.pathreplacing
                matrix,
                positioning
                }

\DeclareAcronym{ADMM}{
    short = ADMM,
    long = alternating direction method of multipliers,
    list = Alternating Direction Method of Multipliers,
    tag = abbrev
}

\DeclareAcronym{AoA}{
    short = AoA,
    long = angle-of-arrival,
    list = Angle-of-Arrival,
    tag = abbrev
}

\DeclareAcronym{SISO}{
    short = SISO,
    long = single-input single-output,
    list = single-input single-output,
    tag = abbrev
}

\DeclareAcronym{MRT}{
    short = MRT,
    long = maximum ratio transmitter,
    list = maximum ratio transmitter,
    tag = abbrev
}

\DeclareAcronym{PDA}{
    short = PDA,
    long = placement delivery array,
    list = placement delivery array,
    tag = abbrev
}

\DeclareAcronym{EE}{
    short = EE,
    long = energy efficiency,
    list = energy efficiency,
    tag = abbrev
}

\DeclareAcronym{MDS}{
    short = MDS,
    long = maximum distance separation,
    list = maximum distance separation,
    tag = abbrev
}

\DeclareAcronym{SIC}{
    short = SIC,
    long = successive-interference-cancellation,
    list = successive-interference-cancellation,
    tag = abbrev
}

\DeclareAcronym{MAC}{
    short = MAC,
    long = multiple-access-channel,
    list = multiple-access-channel,
    tag = abbrev
}

\DeclareAcronym{AoD}{
    short = AoD,
    long = angle-of-departure,
    list = Angle-of-Departure,
    tag = abbrev
}

\DeclareAcronym{BB}{
    short = BB,
    long = base band,
    list = Base Band,
    tag = abbrev
}

\DeclareAcronym{BC}{
    short = BC,
    long = broadcast channel,
    list = Broadcast Channel,
    tag = abbrev
}

\DeclareAcronym{BS}{
    short = BS,
    long = base station,
    list = Base Station,
    tag = abbrev
}

\DeclareAcronym{BR}{
    short = BR,
    long = best response,
    list = Best Response, 
    tag = abbrev
}

\DeclareAcronym{CB}{
    short = CB,
    long = coordinated beamforming,
    list = Coordinated Beamforming,
    tag = abbrev
}

\DeclareAcronym{CC}{
    short = CC,
    long = coded caching,
    list = Coded Caching,
    tag = abbrev
}

\DeclareAcronym{CE}{
    short = CE,
    long = channel estimation,
    list = Channel Estimation,
    tag = abbrev
}

\DeclareAcronym{CoMP}{
    short = CoMP,
    long = coordinated multi-point transmission,
    list = Coordinated Multi-Point Transmission,
    tag = abbrev
}

\DeclareAcronym{CRAN}{
    short = C-RAN,
    long = cloud radio access network,
    list = Cloud Radio Access Network,
    tag = abbrev
}

\DeclareAcronym{CSE}{
    short = CSE,
    long = channel specific estimation,
    list = Channel Specific Estimation,
    tag = abbrev
}

\DeclareAcronym{CSI}{
    short = CSI,
    long = channel state information,
    list = Channel State Information,
    tag = abbrev
}

\DeclareAcronym{CSIT}{
    short = CSIT,
    long = channel state information at the transmitter,
    list = Channel State Information at the Transmitter,
    tag = abbrev
}

\DeclareAcronym{CU}{
    short = CU,
    long = central unit,
    list = Central Unit,
    tag = abbrev
}

\DeclareAcronym{D2D}{
    short = D2D,
    long = device-to-device,
    list = Device-to-Device,
    tag = abbrev
}

\DeclareAcronym{DE-ADMM}{
    short = DE-ADMM,
    long = direct estimation with alternating direction method of multipliers,
    list = Direct Estimation with Alternating Direction Method of Multipliers,
    tag = abbrev
}

\DeclareAcronym{DE-BR}{
    short = DE-BR,
    long = direct estimation with best response,
    list = Direct Estimation with Best Response,
    tag = abbrev
}

\DeclareAcronym{DE-SG}{
    short = DE-SG,
    long = direct estimation with stochastic gradient,
    list = Direct Estimation with Stochastic Gradient,
    tag = abbrev
}

\DeclareAcronym{DFT}{
	short = DFT,
	long = discrete fourier transform,
	list = Discrete Fourier Transform,
	tag = abbrev
}

\DeclareAcronym{DoF}{
    short = DoF,
    long = degrees of freedom,
    list = Degrees of Freedom,
    tag = abbrev
}

\DeclareAcronym{DL}{
    short = DL,
    long = downlink,
    list = Downlink,
    tag = abbrev
}

\DeclareAcronym{GD}{
	short = GD, 
	long = gradient descent,
	list = Gradeitn Descent,
	tag = abbrev
}

\DeclareAcronym{IBC}{
    short = IBC,
    long = interfering broadcast channel,
    list = Interfering Broadcast Channel,
    tag = abbrev
}

\DeclareAcronym{i.i.d.}{
    short = i.i.d.,
    long = independent and identically distributed,
    list = Independent and Identically Distributed,
    tag = abbrev
}

\DeclareAcronym{JP}{
    short = JP,
    long = joint processing,
    list = Joint Processing,
    tag = abbrev
}

\DeclareAcronym{KKT}{
    short = KKT,
    long = Karush-Kuhn-Tucker,
    tag = abbrev
}

\DeclareAcronym{LOS}{
	short = LOS,
	long = line-of-sight,
	list = Line-of-Sight,
	tag = abbrev
}

\DeclareAcronym{LS}{
    short = LS,
    long = least squares,
    list = Least Squares,
    tag = abbrev
}

\DeclareAcronym{LTE}{
    short = LTE,
    long = Long Term Evolution,
    tag = abbrev
}

\DeclareAcronym{LTE-A}{
    short = LTE-A,
    long = Long Term Evolution Advanced,
    tag = abbrev
}

\DeclareAcronym{MIMO}{
    short = MIMO,
    long = multiple-input multiple-output,
    list = Multiple-Input Multiple-Output,
    tag = abbrev
}

\DeclareAcronym{MISO}{
    short = MISO,
    long = multiple-input single-output,
    list = Multiple-Input Single-Output,
    tag = abbrev
}

\DeclareAcronym{MSE}{
    short = MSE,
    long = mean-squared error,
    list = Mean-Squared Error,
    tag = abbrev
}

\DeclareAcronym{MMSE}{
    short = MMSE,
    long = minimum mean-squared error,
    list = Minimum Mean-Squared Error,
    tag = abbrev
}

\DeclareAcronym{mmWave}{
	short = mmWave,
	long = millimeter wave,
	list = Millimeter Wave,
	tag = abbrev
}

\DeclareAcronym{MU-MIMO}{
    short = MU-MIMO,
    long = multi-user \ac{MIMO},
    list = Multi-User \ac{MIMO},
    tag = abbrev
}

\DeclareAcronym{OTA}{
    short = OTA,
    long = over-the-air,
    list = Over-the-Air,
    tag = abbrev
}

\DeclareAcronym{PSD}{
    short = PSD,
    long = positive semidefinite,
    list = Positive Semidefinite,
    tag = abbrev
}

\DeclareAcronym{QoS}{
	short = QoS,
	long = quality of service,
	list = Quality of Service,
	tag = abbrev
}

\DeclareAcronym{RCP}{
	short = RCP,
	long = remote central processor,
	list = Remote Central Processor,
	tag = abbrev
}

\DeclareAcronym{RRH}{
    short = RRH,
    long = remote radio head,
    list = Remote Radio Head,
    tag = abbrev
}

\DeclareAcronym{RSSI}{
    short = RSSI,
    long = received signal strength indicator,
    list = Received Signal Strength Indicator,
    tag = abbrev
}

\DeclareAcronym{RX}{
	short = RX,
	long = receiver,
	list = Receiver,
	tag = abbrev
}

\DeclareAcronym{SCA}{
    short = SCA,
    long = successive convex approximation,
    list = Successive Convex Approximation,
    tag = abbrev
}

\DeclareAcronym{SG}{
    short = SG,
    long = stochastic gradient,
    list = Stochastic Gradient,
    tag = abbrev
}

\DeclareAcronym{SNR}{
    short = SNR,
    long = signal-to-noise ratio,
    list = Signal-to-Noise Ratio,
    tag = abbrev
}

\DeclareAcronym{SINR}{
    short = SINR,
    long = signal-to-interference-plus-noise ratio,
    list = Signal-to-Interference-plus-Noise Ratio,
    tag = abbrev
}

\DeclareAcronym{SOCP}{
	short = SOCP, 
	long = second order cone program,
	list = Second Order Cone Program,
	tag = abbrev
}

\DeclareAcronym{SSE}{
    short = SSE,
    long = stream specific estimation,
    list = Stream Specific Estimation,
    tag = abbrev
}

\DeclareAcronym{SVD}{
	short = SVD,
	long = singular value decomposition,
	list = Singular Value Decomposition,
	tag = abbrev
}

\DeclareAcronym{TDD}{
	short = TDD,
	long = time division duplex,
	list = Time Division Duplex,
	tag = abbrev
}

\DeclareAcronym{TX}{
	short = TX,
	long = transmitter,
	list = Transmitter,
	tag = abbrev
}

\DeclareAcronym{UE}{
    short = UE,
    long = user equipment,
    list = User Equipment,
    tag = abbrev
}

\DeclareAcronym{UL}{
    short = UL,
    long = uplink,
    list = Uplink,
    tag = abbrev
}

\DeclareAcronym{ULA}{
	short = ULA,
	long = uniform linear array,
	list = Uniform Linear Array,
	tag = abbrev
}

\DeclareAcronym{UPA}{
    short = UPA,
    long = uniform planar array,
    list = Uniform Planar Array,
    tag = abbrev
}

\DeclareAcronym{WMMSE}{
    short = WMMSE,
    long = weighted minimum mean-squared error,
    list = Weighted Minimum Mean-Squared Error,
    tag = abbrev
}

\DeclareAcronym{WMSEMin}{
    short = WMSEMin,
    long = weighted sum \ac{MSE} minimization,
    list = Weighted sum \ac{MSE} Minimization,
    tag = abbrev
}

\DeclareAcronym{WBAN}{
	short = WBAN,
	long = wireless body area network,
	list = Wireless Body Area Network,
	tag = abbrev
}

\DeclareAcronym{WSRMax}{
    short = WSRMax,
    long = weighted sum rate maximization,
    list = Weighted Sum Rate Maximization,
    tag = abbrev
}

\usepackage{cleveref}

\begin{document}

\title{Achievable DoF Bounds for Cache-Aided Asymmetric MIMO Communications}

% Optimizing Achievable Transmission Schemes for Enhanced Degrees of Freedom in Cache-Aided MIMO Wireless Communication

\author{\IEEEauthorblockN{Mohammad NaseriTehrani, MohammadJavad Salehi, and Antti T\"olli%, 
}
\IEEEauthorblockA{
    Centre for Wireless Communications, University of Oulu, 90570 Oulu, Finland \\
    \textrm{E-mail: \{firstname.lastname\}@oulu.fi}}

\thanks{
This work was supported by Infotech Oulu and by the Research Council of Finland under grants no. 343586 (CAMAIDE) and 346208 (6G Flagship).}

%This article was presented in parts at~\cite{NaseriTehrani2023MulticastSystems} and ~\cite{naseritehrani2023low} \textit{(Corresponding~authors:)} are with the Centre for Wireless Communications, University of Oulu, FIN-90014 Oulu, Finland. (e-mail:First-Name.Last-Name@oulu.fi)
}

%Mohammad NaseriTehrani, MohammadJavad Salehi and Antti T\"olli 

%}

\maketitle

%!TeX root = d2d-cc.tex

\begin{abstract}
%"THIS PAPER IS ELIGIBLE FOR THE STUDENT PAPER AWARD."

Integrating coded caching (CC) into multiple-input multiple-output (MIMO) communications can significantly enhance the achievable degrees of freedom (DoF) in wireless networks. This paper investigates a practical cache-aided asymmetric MIMO configuration with cache ratio $\gamma$, where a server equipped with $L$ transmit antennas communicates with $K$ users, each having $G_k$ receive antennas. We propose three content-aware MIMO-CC strategies: the \emph{min-G} scheme, which treats the system as symmetric by assuming all users have the same number of antennas, equal to the smallest among them; the \emph{Grouping} scheme, which maximizes spatial multiplexing gain separately within each user subset at the cost of some global caching gain; and the \emph{Phantom} scheme, which dynamically redistributes spatial resources using virtual or ``phantom'' antennas at the users, bridging the performance gains of the min-$G$ and Grouping schemes. These strategies jointly optimize the number of users, $\Omega$, and the parallel streams decoded by each user, $\beta_k$, ensuring linear decodability for all target users. Analytical and numerical results confirm that the proposed schemes achieve significant DoF improvements across various system configurations. 
\end{abstract}


\section{Introduction}

% The expanding demand for multimedia content, resulting from emerging applications such as mobile immersive viewing and extended reality,
% %Expanding multimedia content and the surge in applications like mobile immersive viewing and extended reality 
% is driving continuous growth in mobile data traffic. %Existing wireless networks are under pressure due to the stringent demands of these applications, which require high throughput and ultra-low latency at the same time.
%The existing wireless network infrastructure faces strain due to the demanding requirements of these applications, necessitating high throughput and ultra-low latency. 
% This has spurred the development of novel techniques, such as coded caching (CC)~\cite{maddah2014fundamental}, which stands out for its intriguing potential of offering a performance boost proportional to the cumulative cache size of all network users. CC ingeniously leverages network devices' onboard memory as a communication resource, especially beneficial for cacheable multimedia content. 
Increasing demand for multimedia content, driven by applications like immersive viewing and extended reality (XR), is leading to continuous growth in mobile data traffic. Coded caching (CC)~\cite{maddah2014fundamental} has emerged as an effective solution, utilizing the onboard memory of network devices as a communication resource, especially beneficial for cacheable multimedia content.
%. CC provides performance gains proportional to the total cache size of all users, making it especially useful for cacheable multimedia content.
The performance gain in CC arises from multicasting well-constructed codewords to user groups of size $t\!+1$, where the CC gain $t$ is proportional to the cumulative cache size of all users. 
Originally designed for single-input single-output (SISO) setups~\cite{maddah2014fundamental}, CC was later shown to be effective in multiple-input single-output (MISO) systems, demonstrating that spatial multiplexing and coded caching gains are additive~\cite{shariatpanahi2018physical}.
%While CC was originally designed for single-input single-output (SISO) setups~\cite{maddah2014fundamental}, later studies showed it could also be used in multiple-input single-output (MISO) systems, demonstrating that spatial multiplexing and coded caching gains are additive~\cite{shariatpanahi2018physical}. 
This is achieved by serving 
multiple groups of users simultaneously with multiple multicast messages and suppressing the intra-group interference by beamforming.
%multicasting techniques can be used to serve groups of multiple users, each with correlated requests. %Multicasting uses a single beamformer for a group of users, resulting in a higher bandwidth and transmission rate proportional to transmitter spatial multiplexing gain $L$.%the number of TX-antennas $L$.
Accordingly, in a MISO-CC setting with $L$ Tx antennas, $t\!+ L$ users can be served in parallel, and the so-called degree-of-freedom (DoF) of $t+\!L$ is achievable~\cite{shariatpanahi2016multi,shariatpanahi2018physical,lampiris2021resolving,salehi2020lowcomplexity}. 
%In other works, authors in~\cite{tolli2017multi,tolli2018multicast} discussed how multi-group multicast optimized beamformers could improve the performance of MISO-CC schemes in the finite signal-to-noise ratio (SNR) regime. In the same works, the spatial multiplexing gain and the number of partially overlapping multicast messages were flexibly adjusted for a trade-off between design complexity and finite-SNR performance. %in a trade-off for improved design complexity and finite-SNR performance. 

While MISO-CC has been well-studied in the literature, applying CC in multiple-input multiple-output (MIMO) setups has received less attention. In~\cite{cao2017fundamental}, the optimal DoF of cache-aided MIMO networks with three transmitters and three receivers were studied, and in~\cite{cao2019treating}, general message sets were used to introduce inner and outer bounds on the achievable DoF of MIMO-CC schemes. More recently, we studied low-complexity MIMO-CC schemes for single-transmitter setups in~\cite{salehi2021MIMO}, and showed that with $G$ antennas at each receiver, if $\frac{L}{G}$ is an integer, the single-shot DoF of $Gt+L$ is achievable. 
 %We also explored 
%unicast and multicast beamforming strategies for improving the finite-SNR performance of MIMO-CC systems in~\cite{salehi2023multicast}, and designed a high-performance transmission strategy for MIMO-CC setups in~\cite{naseritehrani2024multicast} by formulating the problem of maximizing the symmetric rate w.r.t transmit covariance matrices of the multicast signals. 
In~\cite{naseritehrani2024multicast, tehrani2024enhanced,naseritehrani2024cacheaidedmimocommunicationsdof}, we proposed an improved achievable single-shot DoF bound for MIMO-CC systems, going beyond the DoF of $Gt+L$ explored in~\cite{salehi2021MIMO}. We also designed a high-performance transmission strategy for MIMO-CC setups in~\cite{naseritehrani2024multicast, naseritehrani2024cacheaidedmimocommunicationsdof} by formulating the problem of maximizing the symmetric rate w.r.t transmit covariance matrices of the multicast signals. 
%Unicast and multicast beamforming strategies for improving the finite-SNR performance of MIMO-CC schemes are also considered in~\cite{salehi2022multicast}.
%\textbf{This paragraph should be changed: first, we need to recap the contribution of the latest ICASSP paper, and denote that an optimized DoF bound is given there, but without any proof. And in this paper, we are providing the proof to that theorem.}
%Recently, a flexible CC scheme for MIMO setups is proposed in~\cite{NaseriTehrani2023MulticastSystems} by optimizing the number of users served in each transmission to maximize the achievable DoF.  This work proposed a high-performance beamformer design for MIMO-CC setups by formulating the problem of maximizing the symmetric rate w.r.t transmit covariance matrices for the multicast signals.%, which was non-convex problem and then solved using SCA. 
%In this paper, we rigorously prove the DoF bound introduced in~\cite{naseritehrani2024multicast}, thus solidifying the theoretical foundations of our approach. Additionally, we provide illustrative examples and numerical results, comparing our techniques against baseline schemes to demonstrate their practical significance and showcase their superior performance and efficiency.
%However, in practice, users often have varying antenna configurations 
%All these works, however, assume the same number of antennas at each receiver, which may not be the case in practice.
%varying number of antennas affects in capabilities such as latency, data rates, and quality of experience (QoE), which depend on the provided service. 
All aforementioned works, however, assume the same number of antennas at each receiver, which may not be the case in practice. For example, devices in 5G NR, ranging from high-performance smartphones to low-power IoT devices, are designed to meet diverse requirements--such as latency, data rates, and QoE--resulting in category-specific asymmetric antenna configurations and capabilities~\cite{dahlman20205g,salehi2022enhancing}.
%Asymmetric numbers of antennas affect capabilities such as latency, data rates, and quality of experience~(QoE), depending on the provided service. %The variation in the number of antennas is often due to hardware limitations, which, in turn, affect capabilities such as latency, data rates, and quality of experience (QoE), depending on the provided service.
Despite its practical significance, the achievable single-shot DoF for scenarios where users are equipped with asymmetric antennas has remained unexplored in the MIMO-CC literature. %This gap motivates our work, where we propose solutions to enhance the achievable DoF in asymmetric MIMO-CC configurations.
%In practice, users have varying antenna configurations due to differences in capabilities, such as latency, data rates, and quality of experience (QoE), depending on the provided service. However, the achievable single-shot DoF for setups with users served by asymmetric antennas remains unexplored in the MIMO-CC literature.
  %In practice, users have varying antennas due to their different capabilities, ranging from latency and data rates to quality of experience (QoE), based on the service provided. However, the achievable DoF of users Served with an asymmetric number of antennas has not been well explored in MIMO-CC literature.
%For instance, emergency responders using AR glasses with multiple antennas (high QoS) require low latency and high reliability (critical), while commuters using AR applications for real-time updates on schedules and routes with fewer antennas (moderate QoS) focus on timely updates and reliable connectivity (routine).

%This paper investigates asymmetric MIMO-CC communication and introduces three innovative delivery schemes to enhance the achievable DoF.
To address this gap, we introduce three delivery schemes aimed at enhancing the achievable DoF in asymmetric MIMO-CC communication scenarios. First, we consider the \emph{min-G} scheme, which treats the system as a symmetric setup, assuming all users have the same number of antennas, equal to the minimum among them. 
% Despite its simplicity, the {min-G} scheme achieves considerable DoF gains across a broad range of design parameters and network configurations. 
Next, we introduce the \emph{grouping} strategy, which divides users into subsets based on their number of antennas and serves each subset in orthogonal dimensions. %This scheme demonstrates high efficiency in another range of design parameters and network regions. 
These two strategies establish a trade-off: the {min-$G$} scheme prioritizes a larger cumulative cache size at the expense of some spatial multiplexing capability, whereas the {grouping} scheme maximizes spatial multiplexing gain across all user subsets while compromising some of the global caching gain. Moreover, the {min-$G$} and grouping schemes are complementary, each enhancing performance within specific ranges of design parameters and network configurations. %Moreover, the {min-G} and grouping schemes are complementary, achieving significant DoF gains across distinct ranges of design parameters and network configurations. Together, they ensure efficient performance under a wide variety of operating conditions. 
% These two strategies establish a trade-off: the {min-G} scheme offers a larger cumulative cache size, while the {grouping} scheme achieves greater spatial multiplexing gain across all subsets, except the one with the fewest antennas.
% The {min-G} and phantom schemes complement each other by achieving significant DoF gains across distinct ranges of design parameters and network configurations, ensuring efficient performance under varying conditions.
% Despite its simplicity, the {min-G} scheme achieves considerable DoF gains across a broad range of design parameters and network configurations. This scheme demonstrates high efficiency in another range of design parameters and network regions.
To bridge this trade-off and unify the benefits of both strategies, we propose the \emph{phantom} MIMO-CC scheme, which extends the achievable DoF of the min-$G$ scheme across a broader range of design and network parameters by introducing virtual or ``phantom'' antennas at the users, dynamically reconfiguring spatial resources to fine-tune user counts and received streams, while accommodating asymmetric transmissions within each interval.
%by reconfiguring spatial resource allocation by introducing virtual or "phantom" antenna users, enabling dynamic fine-tuning of both user count and received streams to enhance wireless communication performance.
%by addressing a broader range of design parameters and network configurations.
 This scheme surpasses the DoF performance of the min-$G$ strategy 
%under more diverse conditions 
while complementing the grouping scheme's performance in specific scenarios.
We provide closed-form expressions of the attainable DoF of all the proposed schemes, and compare their DoF values through numerical results. 
An extended version of this paper with detailed examples of proposed delivery schemes is available in~\cite{naseritehrani2025achievabledofboundscacheaided}.

\emph{Notations.} Bold upper- and lower-case letters signify matrices and vectors, respectively. Calligraphic letters denote sets, ${|\CK|}$ denotes set size of $\CK$, ${\CK \backslash \CT}$ represents elements in $\CK$ excluding those in ${\CT}$. %$A_{(j)}$ and $A_k$ denotes group-specific and user-specific design parameters for the \( j \)-th user set and the \( k \)-th user, respectively.
%  $\mathsf B$ is the collection of sets, with its size given as $|\mathsf B|$. Additional notations are introduced as needed.{(k)}

\section{System Model}
\label{section:sys_model}
%\subsection{Network Setup}
We consider a MIMO setup similar to Figure~\ref{fig:ISIT_sysm}, where a single BS with \( L \) transmit antennas serves \( K \) cache-enabled multi-antenna users. Each user~$k \in [K]$ has $G_k$ receive antennas,\footnote{In fact, \(L\) and 
%\(G_{(j)}\)
\(G_k\)
denote attainable spatial multiplexing gains, upper-bounded by the physical antenna counts, channel ranks, and RF chain limits; ``antenna count'' is used for simplicity.
%\(L\) and \(G_j\) denotes spatial multiplexing gains, possibly below antenna counts due to channel rank or RF chain limits; "antennas" is used for simplicity.
}  
%An example system model with \( J \)=2 is shown in Figure~\ref{fig:ISIT_sysm}.
%
and a cache memory of size $MF$ data bits. The users request files from a library $\CF$ of $N$ files, each with the size of $F$ bits. Consequently, the cache ratio at each user is defined as $\gamma = \frac{M}{N}$. 
%Without loss of generality, a normalized data unit is assumed, and $F$ is ignored in the ensuing notations. 
%Unless explicitly stated otherwise, the following content in this section applies to all proposed schemes, except for the grouping scheme detailed in Sec.~\ref{subsec: grouping} due to splitting... . 
% Coded caching gain represents how many copies of the file library could be stored in the cache memories, either within each user set ($j\in \CJ$) deployed in the proposed scheme grouping, defined as $t_j \equiv K_j \gamma$ or across all user sets utiled in the proposed schemes min-G and phantom, which is defined as $t = \sum_{j\in \CJ}t_j = \sum_{j\in \CJ}K_j \gamma = K \gamma$, where $K = \sum_{j\in \CJ} K_j$.

\begin{figure}[t]
    %\begin{subfigure}{0.5\columnwidth}
        \centering
        \hspace{-.2cm}\includegraphics[height = 4.5cm]{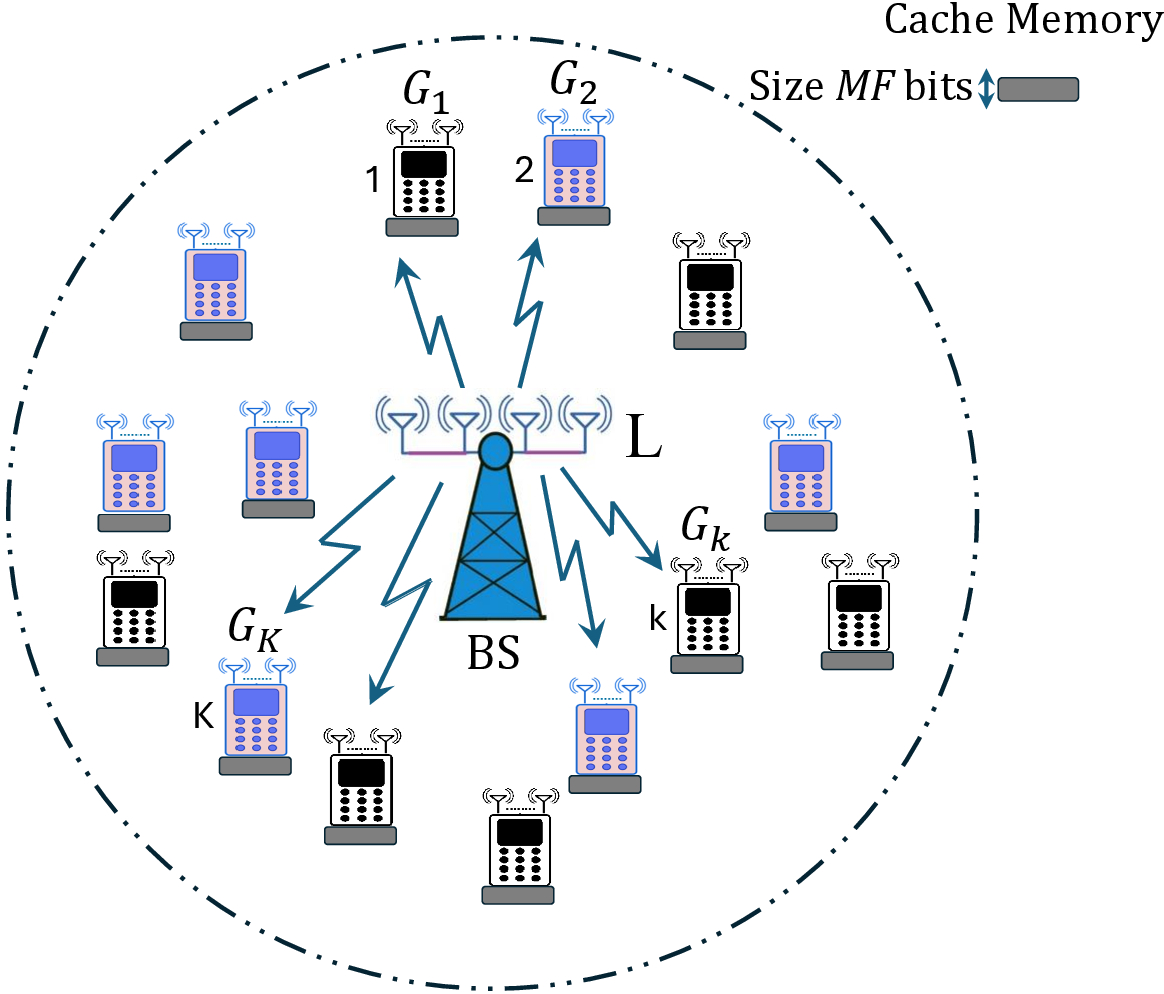} \label{fig:receiver_Prop}
        \vspace{-2mm}
    \caption{Asymmetric MIMO-CC downlink serving a target UE set \(\CK\) of size \(K\), where each UE has the same cache size \(MF\) bits. Different colors represent UEs with varying capabilities.}\label{fig:ISIT_sysm}
\end{figure}

The system operation consists of two phases: cache placement and delivery. In the placement phase, each file in the library is split into a number of smaller \textit{subfiles}, and the cache memory of each user is filled with a subset of subfiles, 
% For each user subset $\CU_{(j)}$ and each user ${k} \in\CU_{(j)} $, we split each file $W \in \CF$ into 
% %$\phi_{k}$
% smaller
% \emph{subfiles} $W_{\CP_{k}}$. The subfile count
% %where 
% %$\phi_{k}$
% is equal to either $\binom{K}{K\gamma}$ or $\binom{K_{(j)}}{K_{(j)} \gamma}$, where $\CP_{k}$ denotes either every subset of $\CK$ with size $K\gamma$ or every subset of $\CU_{(j)}$ with size $K_{(j)}\gamma$,
depending on the particular delivery scheme in use, as described in Section~\ref{sec: Proposed schemes}. 
At the beginning of the delivery phase, each user \(k\) reveals its requested file 
%\(W_{k} \in \CF\)
to the server.
% with \(\Theta\) denoting the final subpacketization levels for each user. Each transmission vector then delivers a new set of subpackets in parallel, with the transmission vector length corresponding to $\frac{1}{\Theta}$. %corresponding to the subpacket size $\frac{1}{\Theta}$. 
%
The server then constructs and transmits a set of transmission vectors \(\Bx(s) \in \mathbb{C}^L\), \(s \in [S]\), e.g., during $S$ consecutive time slots, where $S$ is given by the selected delivery scheme. Each transmission vector $\Bx(s)$ delivers parts of the requested data to every user in a subset \(\CK(s) \subseteq \CK\) of target users, where \(|\CK(s)| = \Omega(s)\), and
% , %with \(\Omega(s) \geq t_j+1\), given $j\in(\CJ\cup \{\emptyset\})$, $s \in [S_j]$.
\(\Omega(s)\) values are chosen to maximize a performance metric. 
Depending on the delivery algorithm, in each transmission, each subfile may be split into a different number of smaller \emph{subpackets} to ensure delivering new data in each transmission vector.
We use $f(s)$ to represent the size of a single stream (equivalently, the size of a single subpacket) delivered during transmission $s$.

In general, the sum rate for delivering the missing $F(1-\gamma)$ bits to all $K$ users can be defined as
% $R_{\textrm{sym}} = \frac{K
% %(1-\gamma)
% F}{T_{\mathrm{total}}}$ 
$R_{\mathrm{total}} = \frac{K F (1-\gamma)}{T_{\mathrm{total}}}$
bits per second (bps), where $T_{\mathrm{total}}$ is the total delivery time. %Let us use $R(s)$ to denote the symmetric (per-stream) rate at transmission $s$ and $f(s)$ to show the size of a single stream (or equivalently, the size of a single subpacket) delivered at that transmission. Then, we have $T_{\mathrm{total}} = \sum_{s \in [S]} \frac{f(s)}{R(s)}$. 
Let us use
$\beta_k(s)$ to denote the number of received parallel streams, and $R_{k,l}(s)$ to denote the achievable rate of stream $l \in [\beta_k]$ of user $k$ in transmission $s$. Then, the delivery time for transmission $s$ is given as $T(s) =  \frac{\sum_{k \in \CK(s)} \beta_k f(s)}{\sum_{k \in \CK(s)} \sum_{l=1}^{\beta_k} R_{k,l}(s)}$, and 
%, and $f(s)$ to represent the size of a single stream (or equivalently, the size of a single subpacket) delivered during that transmission. 
%Since both $R(s)$  and $f(s)$ are scaled by the total number of parallel streams delivered during transmission $s$, their scaling factors cancel out. Thus, the total transmission time can be expressed as
\begin{equation}
    \begin{array}{l}
    T_{\mathrm{total}} = \sum_{s \in [S]} T(s) = \sum_{s \in [S]} \frac{\sum_{k \in \CK(s)} \beta_k f(s)}{\sum_{k \in \CK(s)} \sum_{l=1}^{\beta_k} R_{k,l}(s)}.
    \end{array}
\end{equation}
%T_{\mathrm{total}} = \sum_{s \in [S]} T(s) = \sum_{s \in [S]} \max_{k \in \CK(s)} \frac{f(s)}{R_k(s)}.
Clearly, $R_{k,l}(s)$, $T_{\mathrm{total}}$, and $R_{\mathrm{total}}$ are all functions of the transmission SNR.
%
%use $f(s)$ and $R(s)$ to denote the size of each subpacket (bits) and the error-free transmission rate (bps) at transmission $s$. Then, we can write $T_{\mathrm{total}} = \sum_{s \in [S]} \frac{f(s)}{R(s)}$. 
In this paper, we are interested in characterizing the achievable spatial degrees of freedom~(DoF), defined as:
\begin{equation}\label{eq:DoF_ideal}
\begin{array}{l}
    \mathrm{DoF} = \lim\limits_{\textrm{SNR}\rightarrow\infty}\frac{{R}_{\textrm{total}}}{\log \textrm{SNR}} = \frac{KF(1-\gamma)}{\lim\limits_{\textrm{SNR}\rightarrow\infty} \log \textrm{SNR} \times \lim\limits_{\textrm{SNR}\rightarrow\infty} T_{\textrm{total}}
    }.
\end{array}
\end{equation}
%where $C \coloneq \lim_{\textrm{SNR}\rightarrow\infty} \log(\textrm{SNR})$ shows the channel capacity (bits/s/Hz) under ideal conditions. 
Assuming that inter-stream interference can be completely removed by zero-forcing~(ZF) precoders, as $\textrm{SNR} \to \infty$, each stream can be delivered at a rate approaching $\lim_{\textrm{SNR} \to \infty} R_{k,l}(s) = C \coloneq \lim_{\textrm{SNR}\rightarrow\infty} \log(\textrm{SNR})$, for all $s \in [S]$, $k \in \CK(s)$, and $l \in [\beta_k]$. Hence, $\lim_{\textrm{SNR}\rightarrow\infty} T_{\textrm{total}} = \sum_{s \in [S]} \frac{f(s)}{C}$. Substituting this into~\eqref{eq:DoF_ideal} and simplifying, we get:
\begin{equation}
    \begin{array}{l}\label{eq:DoF_ideal1}
    \mathrm{DoF} = \frac{K F(1-\gamma)}{\sum_{s\in [S]}f(s)}.
    \end{array}
\end{equation}
Now, as each user $k \in \CK(s)$ receives $\beta_k(s)$ parallel streams in transmission $s$, and all the missing subpackets of users are delivered after all $S$ transmissions, we have $KF (1-\gamma) = \sum_{s \in [S]} f(s) %\Omega(s) 
\sum_{k \in \CK(s)} \beta_k(s)$, and~\eqref{eq:DoF_ideal1} can be re-written as
\begin{equation}\label{eq:asym_DoF_phantom_version}
\begin{array}{l}
    \mathrm{DoF} = \frac{\sum_{s \in [S]} f(s) \cdot \sum_{k \in \CK(s)}\beta_k(s)}{\sum_{s \in [S]} f(s)}.
\end{array}
\end{equation}
Depending on the particular delivery scheme, all the users may receive the same number of streams per transmission, i.e., $\beta_k(s) = \beta(s)$ for all $k \in \CK(s)$. In this case, equation~\eqref{eq:asym_DoF_phantom_version}
%On the other hand, if all the users receive the same number of streams per transmission, i.e., $\beta_k(s) = \beta(s)$ for all $k \in \CK(s)$, then~\eqref{eq:asym_DoF_phantom_version} 
can be further simplified to
\begin{equation}\label{eq:asym_DoF}
\begin{array}{l}
    \mathrm{DoF} = \frac{\sum_{s \in [S]} f(s) \cdot \beta(s) \cdot \Omega(s)}{\sum_{s \in [S]} f(s)},
\end{array}
\end{equation}
which can also be interpreted as the weighted sum (by the factor of subpacket length $f(s)$) of per-transmission DoF values (total number of parallel streams, i.e., $\beta(s)\Omega(s)$).
%Equations~\eqref{eq:asym_DoF_phantom_version} and~\eqref{eq:asym_DoF} are used throughout the rest of the paper to calculate the DoF of various schemes.

%On the other hand, we design a delivery scheme to maximize the symmetric rate, defined as $R_{sym} = \nicefrac{K\cdot F}{T_{\mathrm{total}}}$ (files/second). % for $j\in(\CJ\cup \{\emptyset\})$. Here, the total delivery time is given by \( T_{\mathrm{total}} = \sum_{s \in [S]} T(s) \), where the transmission time for \( \mathbf{x}(s) \) is \( T(s) = \nicefrac{f(s)}{R(s)} \) (seconds). In this expression, $f(s)$ (files) is inversely proportional to the final subpacketization level, including splitting factors from both placement and delivery phases, and corresponds to the transmission vector length. \( R(s) \) (files/second) denotes the max-min transmission rate of \( \mathbf{x}(s) \) to ensure successful decoding for all users \( k \in \mathcal{K}(s) \), \( s \in [S] \).
\begin{rem}
    %\textbf{Symmetric Case%simplifies to $\beta \cdot \Omega$
%.}
\label{section:DoF}
To the best of our knowledge, the best DoF bound for any general MIMO-CC setup with a symmetric number of $G$ antennas at all users is given in~\cite{naseritehrani2024cacheaidedmimocommunicationsdof}. The algorithm presented therein assumes the same number of users $\Omega$ served in each transmission, the same number of parallel streams $\beta$ per user, and an equal size of $f = \frac{1}{\Theta}$ per all streams, where $\Theta$ denotes the final subpacketization level. As a result, the DoF relation in~\eqref{eq:asym_DoF} is simplified to $\mathrm{DoF_{\textrm{sym}}}(\Omega,\beta)=\Omega\cdot \beta$, and $\Omega^*$ and $\beta^*$ values maximizing $\mathrm{DoF_{\textrm{sym}}}$ are calculated by solving
\begin{equation}\label{eq:total_DoF}
\begin{array}{l}
  %\displaystyle 
  (\Omega^*, \beta^*) =  \max_{\Omega,\beta}~ \Omega \cdot \beta, \\[1ex]
    %\displaystyle
    \mathrm{s.t.}\:\:{\beta \le \mathrm{min}\bigg({G}, \frac{L \binom{\Omega-1}{K\cdot \gamma}}{1 + (\Omega - K\cdot \gamma-1)\binom{\Omega-1}{K\cdot \gamma}} \bigg).}
\end{array}
\end{equation}\\
In~\cite{naseritehrani2024cacheaidedmimocommunicationsdof}, it is also shown that for any general $\Omega$ and $\beta$ values satisfying the condition in~\eqref{eq:total_DoF}, one can always design a MIMO-CC scheme where in each transmission, every target user can decode all its $\beta$ parallel streams with a linear receiver. Due to these strong properties, we repeatedly use~\eqref{eq:total_DoF} and the delivery algorithm proposed in~\cite{naseritehrani2024cacheaidedmimocommunicationsdof} throughout the rest of this paper.

%In the symmetric antenna case, each interval $s$ serves the same number of users ($\Omega(s) = \Omega$), transmits the same number of streams ($\beta(s) = \beta$), and has equal file sizes per stream ($f(s) = f = \frac{F}{\Theta}$), where $\Theta$ denotes the final subpacketization level. $\mathrm{DoF_{sym}}(\beta,\Omega)=\Omega\cdot \beta$ represents the symmetric DoF necessarily achievable under any symmetric MIMO setup, provided $\beta$ and $\Omega$ satisfy the condition specified in~\cite[Theorem.1]{naseritehrani2024cacheaidedmimocommunicationsdof}. Using Theorem~1, the maximum achievable DoF for the symmetric MIMO-CC scheme is then obtained by solving~\cite[Corollary.1]{naseritehrani2024cacheaidedmimocommunicationsdof}:\\[-2.3mm]
%\begin{equation}\label{eq:total_DoF}
%\begin{array}{l}
%  %\displaystyle 
%  {\mathrm{DoF}_{sym}^*}  (\beta^* , \Omega^*) =  \max_{\beta , \Omega }~ \Omega \cdot \beta, \\[1ex]
%    %\displaystyle
%    \mathrm{s.t.}\:\:{\beta \le \mathrm{min}\bigg({G}, \frac{L \binom{\Omega-1}{K\cdot \gamma}}{1 + (\Omega - K\cdot \gamma-1)\binom{\Omega-1}{K\cdot \gamma}} \bigg),}
%\end{array}
%\end{equation}\\
%where $\beta^*$ and $\Omega^*$ represent the optimal parameters chosen to achieve $\mathrm{DoF}_{sym}$, used as a basis in developing the min-G and grouping delivery schemes.
%The functionin~\eqref{eq:total_DoF} is non-monotonic in $\beta$, so maximum DoF may not occur at $\beta = G$.
\end{rem}
%Throughout the rest of this paper, to find out $\Omega(s)$ and $\beta(s)$ values for each transmission, we frequently use the results obtained in the symmetric MIMO-CC, as to the best of our knowledge, these results offer the best achievable DoF bounds for symmetric MIMO-CC systems in the most general case.

 \color{black}

\section{Asymmetric MIMO-CC Delivery Schemes}
\label{sec: Proposed schemes}
%\subsection{min-\(G\) Scheme}
\subsection{\texorpdfstring{min-$G$ Scheme}{min-G Scheme}}

\label{subsec:minG}
% One way to treat asymmetric MIMO-CC delivery is the min-G scheme. In the min-G scheme, asymmetric MIMO-CC is transformed into symmetric MIMO-CC by setting the effective spatial multiplexing gain as
% \begin{equation}\label{eq:min_G}
%  \begin{array}{l}
%  G_{\mathrm{min-G}} =  \min_{j\in\CJ}(G_j),
% \end{array}
% \end{equation}
%  prioritizing users with lower rank constraints across all user sets. Benefiting larger cumulative cache size at
% the expense of some spatial multiplexing capability, CC-gain is \( t \equiv K \gamma \). 

The min-$G$ scheme transforms the asymmetric MIMO-CC system into a symmetric setup by considering the effective spatial multiplexing gain of
\begin{equation}\label{eq:min_G}
 \begin{array}{l}
% \check{G} = \min_{j\in\CJ} G_{(j)}
 \check{G} = \min_{k\in\CK} G_{k}
 %G_{\mathrm{\min_G}} =  \min_{j\in\CJ}(G_{(j)}),
\end{array}
\end{equation}
for all users, and applying the symmetric MIMO-CC scheme of~\cite{naseritehrani2024cacheaidedmimocommunicationsdof} for both placement and delivery phases. The main idea behind this scheme is to maximize the cumulative cache size of users in the CC delivery session, i.e., to maximize the CC gain by trading off the spatial multiplexing capability of a subset of users. 

\textbf{Placement phase.} Each file $W \in \CF$ is split into $\check{\vartheta} = \binom{K}{K \gamma}$ subfiles $W_{\CP}$, where $\CP$ can denote every subset of the users with size $K\gamma$. Then, at the cache memory of each user~$k$, we store $W_{\CP}$ for all $W \in \CF$ and $\CP \ni k$.

\textbf{Delivery phase.} We first find the optimal number of users served per transmission, $\Omega_{\check{G}}^*$, and the optimal number of streams per user per transmission, $\beta_{\check{G}}^*$, by solving
\begin{equation}\label{eq:optimal_omega_beta_minG}
\begin{array}{l}
\Omega_{\check{G}}^*,\beta_{\check{G}}^* = \arg\max_{\Omega,\beta} \Omega \cdot \beta \\
s.t. \qquad {\beta \le \min\bigg({\check{G}}, \frac{L \binom{\Omega-1}{K\gamma}}{1 + (\Omega- K\gamma - 1)\binom{\Omega-1}{K \gamma}} \bigg).}
\end{array}
\end{equation}
Then, we follow the delivery algorithm in the proof of Theorem~1 in~\cite{naseritehrani2024cacheaidedmimocommunicationsdof}, which involves splitting each subfile $W_{\CP}$ into $\check{\phi} = \binom{K-K\gamma-1}{\Omega^*_{\check{G}}-K\gamma-1} \beta^*_{\check{G}}$ subpackets $W_{\CP}^q$, and creating 
\begin{equation}\label{eq:tran_cnt_minG1}
 \begin{array}{l}
 \check{S} = \binom{K}{\Omega^*_{\check{G}}} \binom{\Omega^*_{\check{G}}-1}{K\gamma}
\end{array}
\end{equation}
transmission vectors, each delivering $\beta^*_{\check{G}}$ parallel streams to every user within a target subset of $\Omega^*_{\check{G}}$ users.

\textbf{DoF analysis.} In the $\textrm{min}-G$ scheme, we have $S = \check{S}$, as defined in~\eqref{eq:tran_cnt_minG1}. For each $s \in [\check{S}]$, we have $f(s) = \nicefrac{F}{\displaystyle\check{\vartheta} \check{\phi}}$, $\Omega(s) = \Omega^*_{\check{G}}$, and $\beta(s) = \beta^*_{\check{G}}$. Substituting these values in~\eqref{eq:asym_DoF} and simplifying the equations, the DoF of the min-$G$ scheme is simply given as
\begin{equation}\label{eq:DoF_minG}
 \begin{array}{l}
 \textrm{DoF}^{*}_{\check{G}} = \Omega^*_{\check{G}} \cdot \beta^*_{\check{G}}.
\end{array}
\end{equation}
In Section~\ref{section:Simulations}, we see that despite its simplicity, the min-$G$ scheme provides solid performance for a wide range of network parameters due to its maximal CC gain.

% and $\beta= G_1$. 

% \begin{equation}\label{eq:total_DoF_min_G}
% \begin{array}{l}
%     \mathrm{DoF}_{\mathrm{min-G}}  (\beta^* , \Omega^*) = \beta^* \times \Omega^* \\[2mm]\hspace{0.5mm}=  \min(G_{K_1}, G_{K_2}) \times \Big(\frac{\displaystyle L}{\displaystyle\min(G_{K_1}, G_{K_2})}+(K_1 + K_2) \cdot \gamma\Big) \\[3mm]\hspace{0.5mm}= G_{K_1}\cdot t + L,
% \end{array}
% \end{equation}
% where $\Omega^* = \frac{L}{G_{K_1}}+K \cdot \gamma $ and $\beta^*= G_{K_1} $ based on Remark~\ref{remark1-Dof}.  

\subsection{Grouping Scheme}
\label{subsec: grouping}
% With the grouping scheme, we treat each group of users \(\CU_{(j)}\) separately. The main rationale behind this scheme is to maximize the effect of spatial multiplexing gain for group, albeit with the cost of decreasing the cumulative cache size (i.e., the CC gain).

With the grouping scheme, we treat each set of users with varying capabilities differently. 
In this regard, the user set $\CK$ is divided into \( J \) disjoint subsets \(\CK_{(j)}\), $j \in \CJ = [J]$, such that \(|\CK_{(j)}| = K_{(j)}\) (i.e., $\sum_{j \in \CJ} K_{(j)} = K$ and $\bigcup_{j \in \CJ} \CK_{(j)} = \CK$). Each user ${k} \in$ \(\CK_{(j)}\) has \(G_{k} = G_{(j)} \) receive antennas, and $G_{(1)} < G_{(2)}$ $< \cdots < G_{(J)}$.
The main rationale behind this scheme is to maximize the effect of spatial multiplexing gain for each group, albeit with the cost of decreasing the cumulative cache size (i.e., the CC gain).

%The user set $\CK$ is divided into \( J \) disjoint subsets \(\CK_{(j)}\), $j \in \CJ = [J]$, such that \(|\CK_{(j)}| = K_{(j)}\) (i.e., $\sum_{j \in \CJ} K_{(j)} = K$ and $\bigcup_{j \in \CJ} \CK_{(j)} = \CK$). Each user ${k} \in$ \(\CK_{(j)}\) has \(G_{k} = G_{(j)} \) receive antennas, where $G_{(1)} < G_{(2)}$ $< \cdots < G_{(J)}$.

\textbf{Placement phase.}
Data placement is done in $J$ steps, where at each step $j \in \CJ$, cache memories of users in group $\CK_{(j)}$ are filled with data. More particularly, at each step $j$, every file $W \in \CF$ is first split into $\vartheta_{(j)} = \binom{K_{(j)}}{K_{(j)}\gamma}$ subfiles $W_{\CP}$, where $\CP$ can denote every subset of the user group $\CK_{(j)}$ with size $K_{(j)} \gamma$. Then, each user $k \in \CK_{(j)}$ stores $W_{\CP}$, for all $W \in \CF$ and $\CP \ni k$.

\textbf{Delivery phase.} Data delivery is also performed in $J$ orthogonal delivery intervals. At each interval $j \in \CJ$, all the missing parts of data files requested by users in group $\CK_{(j)}$ are delivered. First, we calculate $\Omega^*_{(j)}$ and $\beta^*_{(j)}$ by solving
\begin{equation}\label{eq:optimal_omega_beta_grouping}
\begin{array}{l}
\Omega_{(j)}^*,\beta_{(j)}^* = \arg\max_{\Omega,\beta} \Omega \cdot \beta \\
s.t. \qquad {\beta \le \min\bigg({G_{(j)}}, \frac{L \binom{\Omega-1}{K_{(j)}\gamma}}{1 + (\Omega- K_{(j)}\gamma - 1)\binom{\Omega-1}{K_{(j)} \gamma}} \bigg).}
\end{array}
\end{equation}
Then, for the users in $\CK_{(j)}$, we follow the delivery algorithm in the proof of Theorem~1 in~\cite{naseritehrani2024cacheaidedmimocommunicationsdof}. This requires splitting each subfile $W_{\CP}$ into ${\phi}_{(j)} = \binom{K_{(j)}-K_{(j)}\gamma-1}{\Omega^*_{(j)}-K_{(j)}\gamma-1} \beta^*_{(j)}$ subpackets $W_{\CP}^q$, and creating 
\begin{equation}\label{eq:tran_cnt_minG}
 \begin{array}{l}
 {S}_{(j)} = \binom{K_{(j)}}{\Omega^*_{(j)}} \binom{\Omega^*_{(j)}-1}{K_{(j)}\gamma}
\end{array}
\end{equation}
transmission vectors, each delivering $\beta^*_{(j)}$ parallel streams to every user within a target subset of $\Omega^*_{(j)}$ users.

\textbf{DoF analysis.} 
Following the delivery algorithm, for each transmission $s$ in step $j \in \CJ$, we have $f(s) = \nicefrac{F}{\displaystyle \vartheta_{(j)} \phi_{(j)}}$, $\Omega(s) = \Omega^*_{(j)}$, and $\beta(s) = \beta^*_{(j)}$.
Also, as the transmission vectors in step $j$ deliver every missing part of all users in $\CK_{(j)}$, we can write
\begin{equation}
    S_{(j)}\Omega^*_{(j)}\beta^*_{(j)} \times \nicefrac{F}{\vartheta_{(j)} \phi_{(j)}} = K_{(j)}(1-\gamma) F.
\end{equation}
Substituting these into~\eqref{eq:asym_DoF}, the DoF of the grouping scheme, denoted by $\mathrm{DoF}_{\CJ}^*$, can be written as
\begin{equation}\label{eq:total_DoF_grouping}
%\begin{array}{l}
   \hspace{-3mm} \mathrm{DoF}_{\CJ}^* =  \frac{\sum_{j\in \CJ} S_{(j)} \frac{F}{\vartheta_{(j)} \phi_{(j)}}\Omega^*_{(j)}\beta^*_{(j)}
    }{ \sum_{j\in\CJ} S_{(j)} \frac{F}{\vartheta_{(j)} \phi_{(j)}}
    }=  
    \frac{1
    }{ \sum_{j\in\CJ}  \frac{ K_{(j)}/K}{ \Omega^*_{(j)} \beta^*_{(j)}}
    }, %\\   
%\end{array}
\end{equation}
\subsection{Phantom Scheme}
\label{subsec: phantom}
Another approach to designing the asymmetric MIMO-CC delivery scheme is the \emph{Phantom} method, which serves as a bridge between the min-$G$ and grouping schemes. This method maximizes the coded caching (CC) gain while leveraging the excess multiplexing gain of users with more antennas than the minimum value, which would otherwise be wasted in the min-$G$ scheme. However, since users with fewer antennas cannot receive as many streams, a portion of their requested data must be delivered through a separate unicasting (UC) phase. 
\textbf{Placement phase.}  Each file $W \in \CF$ is split into $\hat{\vartheta} = \binom{K}{K\gamma}$ packets $W_{\CP}$, where $\CP$ represents all user subsets of size $K\gamma$. Then, user $k \in [K]$ stores a packet $W_{\CP}$ if $\CP\ni k$.

\textbf{Delivery phase.} Data delivery closely follows the delivery algorithm in the proof of Theorem~1 in~\cite{naseritehrani2024cacheaidedmimocommunicationsdof}. However, adjustments are made to make it suitable for the asymmetric case.
 
\noindent\textit{Step 1.} A parameter $G_{(1)} \le \hat{G} \le G_{(J)}$, and two additional parameters $\hat{\Omega}$ and $\hat{\beta}$ are chosen such that 
%linear decodability is guaranteed if all the users have $\hat{G}$ antennas. Following~\cite{naseritehrani2024cacheaidedmimocommunicationsdof}, this could be done by choosing $\hat{\Omega}$ and $\hat{\beta}$ to satisfy:
%\textbf{Linear decodability constraints: already available in the journal paper. We can simply cite that.}
\begin{equation} \label{DoF_bndd}
\begin{array}{l}
    \hat{\beta} \le \min \left(\hat{G}, \big(L-(\hat{\Omega}-K\gamma-1)\hat{\beta}\big)\binom{\hat{\Omega}-1}{K\gamma}\right).
\end{array}
\end{equation}
According to~\cite{naseritehrani2024cacheaidedmimocommunicationsdof}, this selection guarantees that in the original network setup but now with $\hat{G}$ antennas at each user, a delivery scheme can be designed where each transmission delivers $\hat{\beta}$ parallel streams to every user in a subset of users of size $\hat{\Omega}$.
%which has been proved in~\cite{naseritehrani2024cacheaidedmimocommunicationsdof}.

\noindent\textit{Step 2.}
Consider the original network setup with similar cache placement and user requests, but now with $\hat{G}$ antennas in each user. Following a similar procedure as in~\cite{naseritehrani2024cacheaidedmimocommunicationsdof}, a delivery scheme can be designed with $S_{\textrm{MC}} = \binom{K}{\hat{\Omega}} \binom{\hat{\Omega}-1}{K\gamma}$ transmission vectors $\hat{\Bx}_{\textrm{MC}}(s)$, $s \in [S_{\textrm{MC}}]$, each delivering $\hat{\beta}$ parallel streams to every user~$k$ in a subset $\CK_{\mathrm{MC}}(s)$ of size $\hat{\Omega}$. The delivery algorithm is detailed in~\cite{naseritehrani2024cacheaidedmimocommunicationsdof} and is not repeated here for the sake of brevity. However, in a nutshell, each subfile $W_{\CP}$ is split into $\hat{\phi} = \binom{K-K\gamma-1}{\hat{\Omega}-K\gamma-1} \hat{\beta}$ smaller subpackets $W_{\CP}^q$, and the transmission vectors are built as
\begin{equation}
    %\begin{aligned}\label{sig_model_uc}
    \begin{array}{l}\label{sig_model_uc}
        \hat{\Bx}_{\textrm{MC}}(s) = \sum_{k \in \CK_{\mathrm{MC}}(s)}  \sum_{W_{\CP,k}^q \in \hat{\CM}_{k}(s)}   \Bw_{\CP,k}^q W_{\CP,k}^q,
    \end{array}
\end{equation}
where $W_{\CP,k}^q$ denotes the subpacket $W_{\CP}^q$ of the file requested by user~$k$, $\Bw_{\CP,k}^q$ is the ZF transmit beamforming vector nulling out the interference by $W_{\mathcal{P},k}^q$ to every stream decoded by each user $k' \in \CK_{\mathrm{MC}}(s) \backslash (\CP \cup \{k\})$, and $\hat{\CM}_k(s)$ is the set of subpackets $W_{\CP,k}^q$ to be delivered to user~$k$ in transmission $\hat{\Bx}_{\textrm{MC}}(s)$, built such that $\CP \subseteq \CK_{\mathrm{MC}}(s) \backslash \{k\}$ and the number of subpackets with the same subfile index $\CP$ is minimized as much as possible.

%For each subset of users $\CK_{\mathrm{MC}}(s)$ of size $\hat{\Omega}$, we create $\binom{\hat{\Omega}-1}{K\gamma}$ transmission vectors, each delivering $\hat{\beta}$ parallel streams to each user $k \in \CK_{\mathrm{MC}}(s)$.
%\textbf{A quick recap of the delivery algorithm. This is exactly the algorithm in the proof of Theorem 1 in the journal paper.}

%\footnote{Our goal is to study the ``achievability'' of the proposed DoF value, selecting a subpacketization level that generally meets the requirements, though it may not be minimal for some specific parameter combinations.}
% For each \(\CP \in \SfP_{k}\), where the packet index set is defined as \(\SfP_{k} = \{\CP \subseteq \CK(s) \backslash \{k\} \mid |\CP| = K\gamma\}\), select a distinct set of missing subpackets of \(W_{k}\) indexed by \(\CP\) as \(\CN_{\CP,{k}} = \{W_{\CP,{k}}^q \mid W_{\CP,{k}}^q \text{ is not delivered}\}\), ensuring \(|\CN_{\CP,{k}}| = \beta_{k}\). 
% Among the total subpackets \(\beta_{k}\) transmitted per user in each transmission vector \(\Bx(s)\) for \(s \in [\binom{K-1}{\hat{\Omega}-1}\binom{\hat{\Omega}-1}{K\gamma}]\), \(\theta_{{k},\CP}\) subpackets share the same packet index \(\CP\). 
% Let \(\CM_{k}\) denote the set filled up by \(\beta_{k}\) subpackets \(W_{\CP,{k}}^q\) transmitted via \(\Bx(s)\). By appropriately assigning these subpackets, the following theorem guarantees that the received signals per user are linearly decodable
%%%%%%%%%%%%%%%%%%%%%%%%%%%%%%%%%%%%%%%%%%%%%%%%%%%%%%%%%%%%%

\noindent\textit{Step 3.} For every $k \in \CK$, define $\beta_k = \min({G}_k,\hat{\beta})$.
For each $s \in [S_{\textrm{MC}}]$, create a \emph{multicast} transmission vector $\Bx_{\textrm{MC}}(s)$ for the original setup using $\hat{\Bx}_{\textrm{MC}}(s)$ in~\eqref{sig_model_uc} as follows:
\begin{enumerate}
    \item For each $k \in \CK_{\textrm{MC}}(s)$, create $\CM_{k}(s)$ by removing $\hat{\beta} - \beta_k$ subpackets randomly from $\hat{\CM}_k(s)$. Add the removed subpackets to a set $\CM_{\textrm{UC}}$ to be sent later via \emph{unicast} transmissions.
    \item Create $\Bx_{\textrm{MC}}(s)$ as
    \begin{equation}
    %\begin{aligned}\label{sig_model_uc}
    \begin{array}{l}\label{sig_model_phantom}
        \hspace{-3mm}\Bx_{\textrm{MC}}(s) = \sum_{k \in \CK_{\mathrm{MC}}(s)}  \sum_{W_{\CP,k}^q \in {\CM}_{k}(s)}   \Bw_{\CP,k}^q W_{\CP,k}^q.
    \end{array}
\end{equation}
\end{enumerate}
We denote the total number of missing subpackets by $Z$ and the total number of subpackets delivered via multicasting and unicasting by \(Z_{\textrm{MC}}\) and \(Z_{\textrm{UC}}\), respectively.
%Given the total missing subpackets \(F\), \(F_{\textrm{MC}}\) subpackets are delivered over the transmission intervals \(s \in [S_{\textrm{MC}}]\) via \(\Bx_{\textrm{MC}}(s)\), while the remaining $F_{\textrm{UC}} = |\CK_{\textrm{UC}}|$ subpackets are left for the unicasting.
%A portion of the missing subpackets, \(F\), is delivered over \(s \in [S_{\mathrm{MC}}]\) transmission intervals using MIMO-CC to the users in \(\CK_{\mathrm{MC}}(s)\), referred to as \(F_{\mathrm{MC}}\). 
%removing $\hat{\beta} - \beta_k$ subpackets transmitted to each user $k \in \CK_{\textrm{MC}(s)}$.

%\textbf{We need a Lemma here to show that after removing these subpackets and assuming the real number of antennas, linear decodability is still possible.}

% \noindent \textit{Step 4.} 
% Transmit all the subpackets in $\CU_{\textrm{UC}}$ with unicast transmissions across \(s\in [S_{\mathrm{UC}}]\) intervals to users $k\in\CK_{\mathrm{UC}}(s)$. This involves creating
% \begin{equation}
% \begin{array}{l}
%  \label{Phantom_delivery5}
%     {S_{\mathrm{UC}}} = \displaystyle \frac{\displaystyle |\CU_{\textrm{UC}}|}{\displaystyle \mathrm{DoF_{UC}}},%\\[1mm]
% \end{array}
% \end{equation}
% transmission vectors $\Bx_{\textrm{UC}}(s)$, where
% \begin{equation}
% \begin{array}{l}
% \label{Phantom_delivery6}
%     \mathrm{DoF_{UC}} = \displaystyle \min\Big(L, \sum\nolimits_{k \in \CU_{\mathrm{UC}}}\!\!\! G_k\Big).
% \end{array}
% \end{equation}
% where $\bigcup_{s\in[S_{\mathrm{UC}}]} \CK_{\mathrm{UC}}(s)$.
%\textbf{TODO: fix this part. Define $\CU_{\textrm{UC}}$...}.
\noindent \textit{Step 4.}  
Transmit all the remaining subpackets in \(\CM_{\textrm{UC}}\) using $S_{\textrm{UC}}$ unicast transmissions $\Bx_{\textrm{UC}}(s)$, $s \in [S_{\textrm{UC}}]$. Let us use $d(s)$ to denote the number of subpackets delivered with $\Bx_{\textrm{UC}}(s)$. Successful data delivery requires that $\sum_{s \in [S_\textrm{UC}]} d(s) = |\CM_{\textrm{UC}}| = Z_{\textrm{UC}}$. Clearly, we are interested in increasing each $d(s)$ (and decreasing $S_{\textrm{UC}}$) as much as possible. In principle, for any arbitrary subset $\Bar{\CK}$ of users, the number of parallel streams is limited by $\min (L, \sum\nolimits_{k \in \Bar{\CK}} G_k )$. So, the maximum value of $d(s)$ depends on the user indices of the subpackets transmitted by $\Bx_{\textrm{UC}}(s)$. For the sake of simplicity, in this paper, we assume that the subpackets in $\CM_{\textrm{UC}}$ have diverse user indices, such that there exists a perfect partitioning of $\CM_{\textrm{UC}}$ allowing $d(s) = L$ for all $s \in [S_{\textrm{UC}}]$.\footnote{Relaxing this assumption increases the number of unicast intervals, and reduces the DoF of the phantom scheme. However, as most of the packets are delivered via multicasting, the negative effect is marginal.}

\begin{theorem}\label{Th:Linear_Dec}
The transmission vectors $\Bx_{\textrm{MC}}(s)$ in~\eqref{sig_model_phantom} are linearly decodable by each user $k \in \CK_{\textrm{MC}}(s)$.
    %To ensure linear decodability for each user \(k\) across all sets \(\CU_{(j)}\), \(j \in \CJ\), \(\hat{\beta} - \beta_k\) subpackets must be removed from.
\end{theorem}
\begin{proof}
%Steps: 1) write the inequality (9) for $\hat{\beta}$, 2) check what inequality should be satisfied for the non-symmetric case (LHS remains the same, RHS changes to a summation over $\beta_k$), 3) argue that the inequality is automatically satisfied given the first one as $\hat{\beta} \ge \beta_k$.
The proof is based on how $\Bx_{\textrm{MC}}(s)$ is built by removing \(\hat{\beta} - \beta_k\) subpackets of each user from the reference, linearly decodable transmission vector $\hat{\Bx}_{\textrm{MC}}(s)$ for the equivalent symmetric MIMO-CC network setup. Due to lack of space, it is relegated to the extended version of this paper~\cite{naseritehrani2025achievabledofboundscacheaided}. \end{proof}

\textbf{DoF analysis of the phantom scheme.}

% A portion of the missing subpackets, \(F\), is delivered over \(s\in [S_{\mathrm{MC}}]\) transmission intervals via MIMO-CC to users $k\in\CK_{\mathrm{MC}}(s)$, i.e., \(F_{\mathrm{MC}}\), while the remaining subpackets are transmitted over \(s\in [S_{\mathrm{UC}}]\) intervals using MIMO-UC to users $k\in\CK_{\mathrm{UC}}(s)$, i.e., \(F_{\mathrm{UC}}\). Thus, the total transmissions are \(S = S_{\mathrm{MC}} + S_{\mathrm{UC}}\).
 %The transmitted subpackets in the MIMO-CC and MIMO-UC phases are \(F_{\mathrm{MC}}\) and \(F_{\mathrm{UC}}\), respectively.
% \textbf{A Lemma on how the DoF formula can yield the formula used here}
 \begin{lemma}
     The achievable DoF of the phantom scheme could be calculated as the total number of subpackets delivered over both multicast and unicast transmissions, normalized by the total number of transmissions. In other words,
     %the average number of concurrent users served during the delivery phase over both multicast and unicast transmissions. In other words,
%Thus, the DoF is calculated as follows:\normalfont
\begin{equation}
\begin{array}{l}\label{Phantom_delivery61}
     \textrm{DoF}_{\hat{G}} = \displaystyle
     \frac{%\displaystyle 
     Z_{\textrm{MC}}+Z_{\textrm{UC}}}{%
     %\displaystyle 
     S_{\textrm{MC}}+S_{\textrm{UC}}}. \qquad \mathrm{ } \\[0mm]
\end{array}
\end{equation}
 \end{lemma}
 \begin{proof}
According to the delivery algorithm, the size of each subpacket (and so, the value of $f(s)$) is equal to $\nicefrac{F}{\hat{\phi}\hat{\vartheta}}$ during both multicast and unicast transmissions. Using this in~\eqref{eq:asym_DoF_phantom_version} and noting that $S = S_{\textrm{MC}} + S_{\textrm{UC}}$, we can write
%  Let us start with~\eqref{eq:asym_DoF_phantom_version}, and expand it to:
% \begin{equation}\label{eq:asym_DoF_phantom_version2} 
% \begin{array}{l}
%     \textrm{DoF}_{\hat{G}} =\displaystyle \frac{ \sum\limits_{s \in [S_{\textrm{MC}}]} \sum\limits_{k \in \CK_{\textrm{MC}}(s)}\beta_k(s)+\sum\limits\limits_{s \in [S_{\textrm{UC}}]} \sum\limits_{k \in \CK_{\textrm{UC}}(s)}\beta_k(s)}{S_{\textrm{MC}}+S_{\textrm{UC}}}.
% \end{array}
% \end{equation}
\begin{equation}\label{eq:asym_DoF_phantom_version2} 
\begin{array}{l}
    \hspace{-4mm}\textrm{DoF}_{\hat{G}} =\displaystyle \frac{ \sum\nolimits_{s \in [S_{\textrm{MC}}]} \sum\nolimits_{k \in \CK_{\textrm{MC}}(s)}\beta_k(s)+\sum\nolimits_{s \in [S_{\textrm{UC}}]} d(s)}{S_{\textrm{MC}}+S_{\textrm{UC}}}.
\end{array}
\end{equation}
However, the first and the second summations in the numerator of $\textrm{DoF}_{\hat{G}}$ represent the total number of subpackets in multicast and unicast transmissions, which are equal to $Z_{\textrm{MC}}$ and $Z_{\textrm{UC}}$, respectively, and the proof is complete.
%where we initially split \(S\) into \(S_{\textrm{MC}}\) and \(S_{\textrm{UC}}\) in the summations, then substitute the file size \(f(s) = \nicefrac{F}{\hat{\phi}\hat{\vartheta}}\) and remove it from the fraction, leading to~\eqref{eq:asym_DoF_phantom_version2}. Moreover, the total streams across all users over multicast intervals, \(\sum\nolimits_{s \in [S_{\textrm{MC}}]} \sum\nolimits_{k \in \CK_{\textrm{MC}}(s)}\beta_k(s)\), and unicast intervals, \(\sum\nolimits_{s \in [S_{\textrm{UC}}]} \sum\nolimits_{k \in \CK_{\textrm{UC}}(s)}\beta_k(s)\), amount to \(Z_{\textrm{MC}}\) and \(Z_{\textrm{UC}}\), respectively.  
 \end{proof}
%\begin{theorem}[\textbf{DoF analysis}]\label{Th:Phantom} 
\begin{theorem}\label{Th:Phantom} 
%Given the asymmetric MIMO-CC setup with design parameters ${\beta_{\CU_1}}$, $\hat{\beta}$ and $\Omega$ satisfying the linear decodabilty in Theorem~\ref{Th:DoF_LD}, the achievable DoF of the phantom delivery scheme could be 
%tightly
%approximated as:
Assume the user set $\CK$ is divided into \( J \) disjoint subsets \(\CK_{(j)}\), $j \in \CJ = [J]$, such that \(|\CK_{(j)}| = K_{(j)}\), each user ${k} \in$ \(\CK_{(j)}\) has \(G_{k} = G_{(j)} \) receive antennas, and $G_{(1)} < G_{(2)}$ $< \cdots < G_{(J)}$. Also, the parameters $\hat{\beta}$ and $\hat{\Omega}$ of multicast transmissions satisfy the linear decodability condition in~\eqref{DoF_bndd}, and for all unicast transmissions, we have $d(s) = L$. Then, the DoF of the proposed phantom scheme is given by
%Given the asymmetric MIMO-CC setup with design parameters $K_{(j)}$, $\beta_{(j)}$, $j\in\CJ$, $\hat{\beta}$ and $\hat{\Omega} = |\CK_{\textrm{MC}}(s)|$ satisfying the linear decodability condition in Theorem~\ref{Th:Linear_Dec}, the achievable DoF of the phantom delivery scheme is given by:\normalfont
\begin{equation}
\begin{array}{l}\label{Phantom_delivery0}
    \hspace{-3mm} \textrm{DoF}_{\hat{G}}= \displaystyle 
     \frac{\displaystyle    \hat{\Omega} \times \hat{\beta}}{\displaystyle  1 + \frac{
     %\displaystyle 
     \hat{\Omega}}{\displaystyle K L
} \sum\nolimits_{j\in{
\hat{\CJ}}} \big((\hat{\beta} - \beta_{(j)}) K_{(j)}\big) },  \\[3mm]
\end{array}
\end{equation}
where \(\hat{\CJ} \!= \{ j \! \in \!\CJ \!\mid\! G_{(j)} \!< \!\hat{G} \}\), and ${\beta}_{(j)}\! = \! \beta_k$ for any $k\! \in\! \CK_{(j)}$.
\end{theorem}
\begin{proof}

% By assuming $\hat{\Omega}$ users out of a total \(K\) users are served by MIMO-CC delivery over consecutive time slots, the number of MIMO-CC delivery transmissions is given by:
% \begin{equation}
% \begin{array}{l}\label{Phantom_delivery1}
%     S_{MC} = \displaystyle \binom{K}{\hat{\Omega}}\cdot \binom{\hat{\Omega}-1}{K\gamma}. \qquad \mathrm{ } \\[1mm]
% \end{array}
% \end{equation}
Following the delivery algorithm, the final subpacketization level of the phantom scheme is $\hat{\vartheta}\hat{\phi}$, and so, the total number of missing subpackets for all users is 
%Assuming the delivery scheme, for any given $t$ and $\hat{\Omega}$, $K$ and the number of parallel streams received at each user set $\CK_{(j)}$, $j\in \CJ$, the total missing subpackets is: 
\begin{equation*}
\begin{array}{l}\label{Phantom_delivery51}
    \hspace{0mm}{Z_{}} = K (1-\gamma) {\hat{\vartheta}}\hat{\phi} = \displaystyle K(1-\gamma)\cdot\binom{K}{K\gamma} \binom{K-K\gamma-1}{\hat{\Omega}-K\gamma-1} \hat{\beta}_{}\\[2.5mm] \hspace{22.2mm}=
    \displaystyle K\cdot\binom{K-1}{\hat{\Omega}-1}\cdot\binom{\hat{\Omega}-1}{K\gamma}\cdot  \hat{\beta}_{}.  
\end{array}
\end{equation*}
%By definition, a portion of size $\gamma$ of these subpackets is cached by the users, and the rest should be delivered through multicast and unicast transmissions. 
Following the delivery algorithm, each user $k \in \CK$ appears in $\binom{K-1}{\hat{\Omega-1}}\binom{\hat{\Omega}-1}{K\gamma}$ multicast transmissions, and in each multicast transmission, it receives exactly ${\beta}_k$ subpackets. So, the total number of subpackets delivered to all users by multicasting can be calculated as
\begin{equation}
\begin{array}{l}\label{Phantom_delivery2}
    % Z_{\textrm{MC}} = \displaystyle K \binom{K-1}{\hat{\Omega}-1} \binom{\hat{\Omega}-1}{K\gamma} \beta_k ,
     Z_{\textrm{MC}} = \displaystyle \binom{K-1}{\hat{\Omega}-1} \binom{\hat{\Omega}-1}{K\gamma} \sum\nolimits_{k\in\CK}\beta_k ,
\end{array}
\end{equation}
which, as $\beta_k = \beta_{(j)}$ for all $k \in \CK_{(j)}$, can be rewritten as
\begin{equation}
\begin{array}{l}\label{Phantom_delivery_MC_count}
    Z_{\textrm{MC}} = \displaystyle \binom{K-1}{\hat{\Omega}-1} \binom{\hat{\Omega}-1}{K\gamma} \cdot\sum\nolimits_{j\in\CJ}\big(\beta_{(j)}\cdot K_{(j)}\big).
\end{array}
\end{equation}
Now, as the rest of the subpackets are delivered via unicasting, we have $Z_{\textrm{UC}} = Z - Z_{\textrm{MC}}$, and hence,
\begin{equation}
\begin{array}{l}\label{Phantom_delivery3}
    \hspace{-4mm} Z_{\textrm{UC}} =
     \displaystyle \binom{K-1}{\hat{\Omega}-1} \binom{\hat{\Omega}-1}{K\gamma} 
     \sum\nolimits_{j\in \hat{\CJ}}\Big((\hat{\beta} - \beta_{(j)}) K_{(j)}\Big).
    \qquad %\mathrm{file}
    %\\%[2mm]
\end{array}
\end{equation}
The last step is calculating $S_{\textrm{UC}}$. As we have assumed $d(s) = L$ for every $s \in [S_{\textrm{UC}}]$ and because all the remaining subpackets are delivered via unicasting, we can simply write $S_{\textrm{UC}} = \frac{Z_{\textrm{UC}}}{L}$. Finally, substituting $Z_{\textrm{MC}}$, $Z_{\textrm{UC}}$, $S_{\textrm{MC}}$, and $S_{\textrm{UC}}$ into the DoF formula in~\eqref{Phantom_delivery61}, we get
\begin{equation}
\begin{array}{l}\label{Phantom_delivery10}
    \textrm{DoF}_{\hat{G}}= \displaystyle 
     \frac{
      \displaystyle
     K\binom{K-1}{\hat{\Omega}-1}  \hat{\beta}}{
     \displaystyle
     \binom{K}{\hat{\Omega}} +   \frac{
      %\displaystyle
     \binom{K-1}{\hat{\Omega}-1}}{\displaystyle L} \sum\nolimits_{j\in  
\hat{\CJ}}\!\!\Big((\hat{\beta} - \beta_{(j)}) K_{(j)}\Big)},\\[2mm] 
% \hspace{11mm} =  
%     \frac{
%     \displaystyle
%     \Omega \cdot \hat{\beta}}{\displaystyle  1 + \frac{
%     %\displaystyle
%     \Omega}{K\cdot L} \cdot \sum\nolimits_{j\in\CJ_{\backslash
% \hat{\CJ}}}\!\! \Big((\hat{\beta} - \beta_j)\cdot K_j\Big) }. \qquad \\[1mm]
\end{array}
\end{equation}
% \begin{equation}
% \begin{array}{l}\label{Phantom_delivery10}
%      \mathrm{DoF_{ph}}= \displaystyle \frac{\displaystyle K\cdot\binom{K-1}{\Omega-1} \cdot \hat{\beta}}{\displaystyle \binom{K}{\Omega} +   \frac{\displaystyle\binom{K-1}{\Omega-1}\cdot K_1\cdot \big(\hat{\beta}-\beta_{\CU_1}
%     \big)}{\displaystyle \min\big(L, \sum\nolimits_{k \in \CU_1}\!\!\! G_k\big)}},\\[2mm] \hspace{11mm} =  %\frac{\displaystyle K\cdot L \cdot \Omega \cdot \hat{G}}{\displaystyle K\cdot L + K_1 \cdot \Omega \cdot \big(\hat{G}-G_{\mathrm
%     %{K_1}}\big)}
%     \frac{\displaystyle  \Omega \cdot \hat{\beta}}{\displaystyle  1 + \frac{\displaystyle K_1 \cdot \Omega \cdot \big(\hat{\beta}-\beta_{\CU_1}\big)}{K\cdot \min\big(L, \sum\nolimits_{k \in \CU_1}\!\! G_k\big)}}. \qquad \\[1mm]
% \end{array}
% \end{equation}
% \begin{equation}
% \begin{array}{l}\label{Phantom_delivery10}
%      \mathrm{DoF_{ph}}= \displaystyle \frac{\displaystyle K\cdot\binom{K-1}{\Omega-1} \cdot \hat{G}}{\displaystyle \binom{K}{\Omega} +   \frac{\displaystyle\binom{K-1}{\Omega-1}\cdot K_1\cdot \big(\hat{G}-G_{K_1}
%     \big)}{\displaystyle L}},\\[2mm] \hspace{11mm} =  %\frac{\displaystyle K\cdot L \cdot \Omega \cdot \hat{G}}{\displaystyle K\cdot L + K_1 \cdot \Omega \cdot \big(\hat{G}-G_{\mathrm
%     %{K_1}}\big)}
%     \frac{\displaystyle  \Omega \cdot \hat{G}}{\displaystyle  1 + \frac{\displaystyle K_1 \cdot \Omega \cdot \big(\hat{G}-G_{\mathrm
%     {K_{1}}}\big)}{K\cdot L}}. \qquad \\[1mm]
% \end{array}
% \end{equation}
and after simplification, the proof is completed.
\end{proof}
\section{Numerical Examples}
\label{section:Simulations}

\begin{table}[!b]%[!htbp]
\vspace{10pt}
\centering
\caption{DoF of Asymmetric MIMO-CC Schemes for $G_{(1)}=2$, $G_{(2)}=4$, $L=12$, $\gamma=0.04$, for $(K_{(1)}$, $K_{(2)})$: a) $(25, 75)$, b) $(50, 50)$, c) $(75, 25)$   }
\vspace{-1mm}
\renewcommand{\arraystretch}{1.3}  % Increasing row height
\fontsize{22}{24}\selectfont % Set font size to 22pt with 24pt line spacing
\vspace{-2mm}
\begin{subtable}{0.5\textwidth}
\caption{}
\resizebox{\textwidth}{!}{ % Resize the table to fit the text width
\begin{tabular}{|c|c|c|c|c|c|c|c|>{\centering\arraybackslash}p{1.5cm}|>{\centering\arraybackslash}p{1.5cm}|}
\hline
\multicolumn{7}{|c|}{$\textrm{DoF}_{\hat{G}}$} & $\textrm{DoF}^{*}_{\check{G}}$ & \multicolumn{2}{|c|}{$\textrm{DoF}_{\CJ}^*$=20.36} \\ 
\hline
\diagbox[dir=NW,width=7.5em, height=2em]{$\hat{\beta}$ \text{or} $\beta_k$}{\raisebox{-.3em}{$\Omega$}} & 5 & 6 & 7 & 8 & 9 & 10 & 10 & 7 & 6 \\ 
\hline
2 & 10 & 12 & 14& 16 & 18 & 20 & 20 & 14 & -- \\ 
3 &  13.58 & 16.00 & 18.33 & 20.57 & -- & -- & -- & --&-- \\ 
4 & 16.55 & 19.20 & 21.68 &-- & -- & -- & -- & -- & 24 \\ 
\hline
\end{tabular}\label{aTable}
}
\end{subtable}

\vspace{0.5mm}
\begin{subtable}{0.5\textwidth}
\caption{ }
\resizebox{\textwidth}{!}{ % Resize the table to fit the text width
\begin{tabular}{|c|c|c|c|c|c|c|c|>{\centering\arraybackslash}p{1.5cm}|>{\centering\arraybackslash}p{1.5cm}|}
\hline
\multicolumn{7}{|c|}{$\textrm{DoF}_{\hat{G}}$} & $\textrm{DoF}^{*}_{\check{G}}$ & \multicolumn{2}{|c|}{$\textrm{DoF}_{\CJ}^*$=17.78} \\ 
\hline
\diagbox[dir=NW,width=7.5em, height=2em]{$\hat{\beta}$ \text{or} $\beta_k$}{\raisebox{-.3em}{$\Omega$}} & 5 & 6 & 7 & 8 & 9 & 10 & 10 & 8 & 5 \\ 
\hline
2 & 10 & 12 & 14& 16 & 18 & 20 & 20 & 16 & -- \\ 
3 &  12.41 & 14.40 & 16.26 & 18 & -- & -- & -- & --&-- \\ 
4 & 14.12 & 16.00 & 17.68 &-- & -- & -- & -- & -- & 20 \\ 
\hline
\end{tabular} \label{bTable}
}
\end{subtable}

\vspace{.5mm}
\begin{subtable}{0.5\textwidth}
\caption{}
\resizebox{\textwidth}{!}{ % Resize the table to fit the text width
\begin{tabular}{|c|c|c|c|c|c|c|c|>{\centering\arraybackslash}p{1.5cm}|>{\centering\arraybackslash}p{1.5cm}|}
\hline
\multicolumn{7}{|c|}{$\textrm{DoF}_{\hat{G}}$} & $\textrm{DoF}^{*}_{\check{G}}$ & \multicolumn{2}{|c|}{$\textrm{DoF}_{\CJ}^*$=17.45} \\ 
\hline
\diagbox[dir=NW,width=7.5em, height=2em]{$\hat{\beta}$ \text{or} $\beta_k$}{\raisebox{-.3em}{$\Omega$}} & 5 & 6 & 7 & 8 & 9 & 10 & 10 & 9 & 4 \\ 
\hline
2 & 10 & 12 & 14& 16 & 18 & 20 & 20 & 18 & -- \\ 
3 &  11.43 & 13.09 & 14.61 & 16 & -- & -- & -- & --&-- \\ 
4 & 12.31 & 13.71 & 14.93 &-- & -- & -- & -- & -- & 16 \\ 
\hline
\end{tabular}\label{cTable}
}
\end{subtable}\label{Table_1}

\end{table}
\begin{table}[!t]%[!htbp]
\centering
\caption{DoF of MIMO-CC Schemes for  $G_{(1)}=2$,$G_{(2)}=4$, $L=12$, $K = 500$ }
\renewcommand{\arraystretch}{1.1}  % Increasing row height
\fontsize{7}{8}\selectfont % Set font size to 10pt with 12pt line spacing
\vspace{-2mm}
\resizebox{0.43\textwidth}{!}{ % Resize the table to fit the text width
\begin{tabular}{|c|c|c|c|c|c|c|c|>{\centering\arraybackslash}p{1.5cm}|>{\centering\arraybackslash}p{1.5cm}|}
\hline
$\gamma$&($K_{(1)}$, $K_{(2)}$) &{$\textrm{DoF}_{\hat{G}}^*$} & $\textrm{DoF}^{*}_{\check{G}}$ & {$\textrm{DoF}_{\CJ}^*$} \\  
\hline
&(10,490) & 180.17 & 112 & 162.86 \\ \cdashline{2-10}
&(20,480) & 156.65 & 112 & 138.78 \\ \cdashline{2-10}
0.1&(30,470) & 138.56 & 112 & 124.48 \\ \cdashline{2-10}
&(40,460) & 124.22 & 112 & 115.02  \\ \cdashline{2-10} 
&(50,450) & 112.57 & 112 & 108.31 \\ \cdashline{2-10} 
&(60,440) & 112 & 112 & 103.30  \\ \hline
&(50,450) & 66.51 & 52 &  58.95 \\ \cdashline{2-10}
&(75,425) & 58.41 & 52 & 52.75 \\ \cdashline{2-10}
0.04&(100,400) & 52.08 & 52 & 48.72 \\ \cdashline{2-10}
&(125,375) & 52 & 52 & 45.91 \\ \hline
&(100,400) & 25.26 & 22 & 23.33 \\ \cdashline{2-10}
0.01&(200,300) & 22 & 22 & 20 \\ \cdashline{2-10}
&(400,100) & 22 & 22 & 19.05 \\ 
%&(175,325) & 52 & 22 & 50.70 \\ 
\hline
\end{tabular}
}\label{Table_2}
\end{table}
Numerical examples are generated for two target user sets, \(\mathcal{K}\), with sizes \(K \in \{100, 500\}\). Users are further divided into \(J = 2\) groups based on their UE capability, with \(G_k \in \{2, 4\}\) antennas. The evaluation is conducted under various combinations of parameters, including \(\gamma\), \(L = 12\), \(G_{(1)} = 2\), \(G_{(2)} = 4\), \(\Omega\), \(\hat{G}\), and \(G_k = G_{(j)}\) for all \(k\), where \(j \in [2]\) identifies the group \(\CK_{(j)}\) to which the \(k\)-th user belongs. The sizes of the user groups are chosen such that the coded caching gain, \(K_{(j)} \cdot \gamma \in \mathbb{N}\), remains an integer for simplicity.

In Table~\ref{Table_1}, we illustrate the impact of the UE group size, associated with two different Rx antenna capabilities, on the achievable Degrees of Freedom (DoF) for each proposed scheme. As observed, the state-of-the-art \(\mathrm{min}-G\) solution 
%represents a special case of the phantom scheme, 
provides the same DoF value as the phantom scheme in a specific setup where $\hat{G}=\check{G}$ and $\hat{\Omega}=\Omega^*_{\check{G}}$.
%$\hat{\beta}$, 
Both the phantom and grouping schemes can exceed the DoF of \(\mathrm{min}-G\) by leveraging varying spatial multiplexing capability of each UE. However, the phantom scheme achieves a greater DoF boost than the grouping scheme, as it simultaneously exploits the maximum available CC-gain and the Rx spatial multiplexing capabilities of all users.    
%, consistent with conditions 1, 2, and 3 in lemma~\ref{lm: DoF_lemma_grouping_min_G}.  
%Also, Table~\ref{aTable}, \ref{bTable}, \ref{cTable}, confirms that $\beta^* = G$ and $\hat{\beta}=\hat{G}$ are the optimal selection for achieving the maximum DoF, holding true across the $\min-G$, grouping, and phantom schemes, respectively. 
 Moreover, Tables~\ref{aTable}--\ref{cTable} demonstrate that the DoF values of the phantom scheme can be fine-tuned %based on demand,
 depending on \(\hat{G}\) and \(\Omega\), while adhering to the linear decodability criterion outlined in Theorem~\ref{Th:Linear_Dec}. 
%Moreover, the phantom DoF values vary based on $\hat{G}$ and $\Omega$ while satisfying the linear decodability criterion defined in Theorem~\ref{Th:DoF_LD}. 
%indicate the optimized achievable DoF~\eqref{eq:total_DoF}.
%%%NOTE
% As can be seen, even the small CC gain of $t=1$ results in a big performance improvement in terms of both the DoF and the symmetric rate, and the performance improvement is much more prominent with larger $G$. To justify this, we recall that our achievable DoF in~\eqref{eq:total_DoF} is larger than $Gt+L$ that was achieved in~\cite{salehi2021MIMO} and scaled linearly with $G$.

In Table~\ref{Table_2}, we assess the achievable DoF values offered by the proposed schemes versus the cache ratio per UE $\gamma$. The analysis highlights the remarkable scalability of the achievable DoF in the proposed designs as both user set sizes %(network scales) 
and the design parameter $\gamma$ (CC-gain) are increased.
In Table~\ref{Table_2}, it is evident that the DoF gap for the phantom and grouping schemes become significantly more pronounced as the user set sizes and \(\gamma\) increase for \(K_{(1)} < K_{(2)}\). Additionally, for larger values of \(\gamma\), the DoF of the phantom and grouping schemes surpasses that of \(\min-G\) when there is a greater disparity between the user set sizes, unlike for smaller values of \(\gamma\). Furthermore, the phantom scheme can exceed the DoF of the grouping scheme by simultaneously effectively harnessing the CC gain and achieving an adjustably higher Rx spatial multiplexing gain per user. On the other hand, for \(K_{(1)} \geq K_{(2)}\), the DoF of the phantom or \(\min-G\) schemes exceeds that of the grouping scheme.
Altogether, the proposed schemes provide remarkable flexibility by considering both network and design parameters. They enable seamless scheduling to switch between schemes, maximizing the DoF for asymmetric MIMO-CC systems while addressing the limitations of individual approaches.

\section{Conclusion}
\label{sec:conclusions}
This paper studies data delivery in cache-aided MIMO communications with an asymmetric number of antennas on the users. Three data delivery schemes are proposed, each offering distinct advantages by maximizing the caching or spatial multiplexing gain or enabling a trade-off between them. Closed-form expressions of the achievable DoF for each delivery scheme are derived, and their performance is compared through numerical analysis. 
\newpage

\newpage

\bibliographystyle{IEEEtran}
\bibliography{conf_short,IEEEabrv,references,whitepaper}
%\bibliography{Bib/conf_short,Bib/IEEEabrv,Bib/references,Bib/whitepaper}
%\newpage
%\newpage
\clearpage
   
\section{Appendix}

\subsection{Proof of Theorem~\ref{Th:Linear_Dec}}
\label{app_A}
%\subsection{Proof of Theorem~\eqref{Th:Linear_Dec}}
%\begin{proof}
We show that after the reception of each ${\Bx}_{\textrm{MC}}(s)$, every user $k \in \CK_{\textrm{MC}}(s)$ can decode $\beta_k$ subpackets of its requested file interference-free and with a linear receiver. 
%
%If users \(k \in \CK_{(j)}\) belong to sets \(j \in (\CJ\backslash\hat{\CJ})\), then the spatial multiplexing gain satisfies \(\hat{G} \leq G_k\) and \(\hat{\beta} \leq G_k\), resulting in each user being allocated \(\beta_k = \hat{\beta} = \min(G_k, \hat{\beta})\). Consequently, no subpackets need to be removed, and linear decodability is ensured as~\eqref{DoF_bndd}. For user sets \(j \in  \hat{\CJ}\), where \(G_k < \hat{G}\), all users \(k \in \CK_{(j)}\) are allocated \(\beta_k \leq \min(G_k, \hat{\beta})\), which requires analysis to ensure linear decodability after removing \(\hat{\beta} - \beta_k \neq 0\) subpackets.
%
Given the same insight outlined in~\cite{naseritehrani2024cacheaidedmimocommunicationsdof},
%First, note that by selecting the set $\CN_{\CP,k}$ with the largest cardinality at every iteration of Step~2,
%the above subpacket selection algorithm minimizes the number of subpackets $W_{\CP}^q(k)$ in $\Bx(s)$ with the packet index $\CP$.
%As a result, as we transmit $\beta_k$ subpackets with each transmission and there exist $S$ distinct packet index sets $\CP$, 
and the fact that $\hat{\beta}-\beta_k$ subpackets of the file requested by each user $k \in \CK_{\textrm{MC}}(s)$ are randomly dropped from $\hat{\Bx}_{\textrm{MC}}(s)$ to build ${\Bx}_{\textrm{MC}}(s)$,
%to ensure linear decodability of each user $k$, 
the maximum number of subpackets $W_{\CP}^q(k)$ in ${\Bx}_{\textrm{MC}}(s)$ with the same subfile index $\CP$ is $\big\lceil \nicefrac{\hat{\beta}}{\binom{\hat{\Omega}-1}{K\gamma}} \big\rceil$ (subpackets with the same subfile index may remain after removal). On the other hand, by definition, the beamformer vector $\Bw_{\CP,k}^q$ has to null out the inter-stream interference caused by $W_{\CP}^q(k)$ to every stream decoded by user $k'$ in $\CK_{\textrm{MC}}(s) \backslash (\CP \cup \{k\})$ (the interference caused to other users is removed by their cache contents).
Let us use $\BU_k \in \mathbb{C}^{G_k \times \beta_k}$ to denote the receive beamforming matrix at user~$k$ for decoding its $\beta_k$ parallel streams. For transmission \(s\), the \emph{equivalent interference channel} of a subpacket \(W_{\CP}^q(k)\) included in \({\Bx}_{\textrm{MC}}(s)\) is defined as:
%For the given transmission $s$, let us define the \emph{equivalent interference channel} of a subpacket $W_{\CP}^q(k)$ included in $\Bx(s)$ as
\begin{equation}
\label{eq:equi_intf_chan}
\begin{array}{l}
\bar{\BH}_{\mathcal{P},k}   = [\BH_{k'}\herm\BU_{k'}]\herm, \quad \forall k' \in \CK_{\textrm{MC}}(s) \backslash (\mathcal{P}\cup \{k\}),
\end{array}
\end{equation}
where \([\;\cdot\;]\) denotes horizontal matrix concatenation. For each subpacket \(W_{\CP,k}^q\) transmitted by \({\Bx}_{\textrm{MC}}(s)\), the interference nulling condition for its respective beamformer \(\Bw_{\CP,k}^q\) requires:
%where $[\;\cdot\;]$ represents the horizontal concatenation of matrices inside the brackets. Then, for each subpacket $W_{\CP,k}^q$ transmitted by $\Bx(s)$, the interference nulling condition of its respective beamformer $\Bw_{\CP,k}^q$ implies that
\begin{equation}\label{eq:beamforming_in_nullspace}
    \Bw_{\CP,k}^q \in \mathrm{Null}(\bar{\BH}_{\CP,k}), %\: \forall \CP \in \SfP(k).
\end{equation}
%Note that for given user index $k$ and packet index $\CP$, the definition in~\eqref{eq:equi_intf_chan} holds for each subpacket $W_{\CP}^q(k)$ transmitted by $\Bx(s)$, irrespective of its subpacket index $q$.
where $\mathrm{Null}(\cdot)$ shows the null space.
Note that for a given user index \(k\) and subfile index \(\CP\), the definition in~\eqref{eq:equi_intf_chan} applies to each subpacket \(W_{\CP}^q(k)\) transmitted by \({\Bx}_{\textrm{MC}}(s)\), regardless of the subpacket index \(q\).
Now, to ensure that each user~$k$ can decode all $\beta_k \leq G_k$ parallel streams $W_{\CP,k}^q$ sent to it with ${\Bx}_{\textrm{MC}}(s)$, it is essential that transmit beamformers $\Bw_{\CP,k}^q$ are linearly independent. %For parallel streams with a different packet index $\CP$, this condition is automatically met since the beamformers are chosen from different null spaces.
%However, 
% The bottleneck to hold that condition appears
% for successfully decoding subpackets with the same packet index $\CP$, the number of such subpackets $\nicefrac{\hat{\beta}}{\binom{\hat{\Omega}}{K\gamma}}$ should be limited by the dimensions of the $\mathrm{Null}(\bar{\BH}_{\CP,k})$. From~\eqref{eq:equi_intf_chan}, $\bar{\BH}_{\CP,k}$ is formed by concatenation of $\hat{\Omega}-K\gamma-1$ matrices $\BH_{k'}\herm\BU_{k'}$, each with dimensions $L \times \beta_{k'}$, and hence, $\bar{\BH}_{\CP,k} \in \mathbb{C}^{\sum_{k' \in \CK(s) \backslash (\mathcal{P}\cup \{k\})} \beta_{k'} \times L}$. 
The bottleneck for successfully decoding subpackets with the same subfile index \(\CP\) arises because the number of such subpackets, \(\lceil\nicefrac{\hat{\beta}}{\binom{\hat{\Omega}-1}{K\gamma}}\rceil\), must be limited by the dimensions of \(\mathrm{Null}(\bar{\BH}_{\CP,k})\). From~\eqref{eq:equi_intf_chan}, \(\bar{\BH}_{\CP,k}\) is constructed by concatenating \(\hat{\Omega} - K\gamma - 1\) matrices \(\BH_{k'}\herm\BU_{k'}\), each of size \(L \times \beta_{k'}\). Thus, \(\bar{\BH}_{\CP,k} \in \mathbb{C}^{\sum_{k' \in \CK_{\textrm{MC}}(s) \backslash (\mathcal{P} \cup \{k\})} \beta_{k'} \times L}\). Using 
$\mathrm{nullity(\cdot)}$ to denote the dimensions of the null space, we can use the rank-nullity theorem~\cite{meyer2023matrix} to write
\begin{equation*}
 \label{eq:rank-null-theorem}
 \begin{array}{l}
\mathrm{nullity(\bar{\BH}_{\CP,k})} =L - \mathrm{rank}(\bar{\BH}_{\mathcal{P},k})  
\\[2mm] \hspace{21.5mm}= L - \sum\nolimits_{k' \in \CK_{\textrm{MC}}(s) \backslash (\mathcal{P}\cup \{k\})} \beta_{k'}.
\end{array}
 \end{equation*}
So, for successful decoding of subpackets with the same subfile index, we should have
\begin{equation}
\label{eq:ceil_ineq_dof}
\begin{array}{l}
     \hat{\beta}  \le \big(L -\sum_{k' \in \CK_{\textrm{MC}}(s) \backslash (\mathcal{P}\cup \{k\})} \beta_{k'}\big)\binom{\hat{\Omega}-1}{K\gamma}.\\[3mm]
\end{array}
\end{equation}
Additionally, to ensure linear decodability for all users, the right-hand side (R.H.S) of~\eqref{eq:ceil_ineq_dof} must be larger than or equal to the R.H.S of~\eqref{DoF_bndd}. This condition is always satisfied because:
\begin{equation}
\label{eq:ceil_ineq_dof1}
\begin{array}{l}
     \sum\nolimits_{k' \in \CK_{\textrm{MC}}(s) \backslash (\mathcal{P}\cup \{k\})} \beta_{k'}  \le  (\hat{\Omega}-K\gamma-1)\hat{\beta}.
\end{array}
\end{equation}
Combined with $\beta_k \le \hat{\beta}$ and the decoding criterion $\beta_k \le G_k$, this guarantees linear decodability for each user. 
% which, together with $\beta_k \le \hat{\beta}$ and the basic decoding criteria of $\beta_k \le G_k$, ensures the linear decodablity per user. 
%in~\eqref{DoF_bndd}.
    %For user sets \(\hat{\CJ} = \{ j \in \CJ \mid \hat{G} \leq G_{(j)} \}\), where the spatial multiplexing gain satisfies \(\hat{G} \leq G_{(j)}= G_k\) and $\hat{\beta}\leq G_k$, the streams allocated per user are \(\beta_k= \beta_{(j)} = \hat{\beta}=\min(G_k,\hat{\beta})\), for $k\in\CU_{(j)}$, $j\in \hat{\CJ}$. This ensures users within \(\hat{\CJ}\) set receives all missing subpackets and are linear decodable. For user sets \(j\in\CJ_{\backslash
 %\hat{\CJ}}\), where \(G_k= G_{(j)} < \hat{G}\), the streams allocated per user must satisfy \(\beta_k = \beta_{(j)} \leq \hat{\beta}\) and \(\beta_k=\beta_{(j)} \leq G_k=G_{(j)}\). Therefore, $\beta_k\leq G_k$ is always than $G_k$ to enable nulling of interfering streams.   
%\end{proof}

\subsection{Transmission vector Design: Example
%in Sections~\eqref{sec: Proposed schemes} and~\eqref{sec: phantom}
}
\label{app_C}

\begin{exmp}
\label{exmp:x0serv}
\normalfont 
In this example, we revisit and evaluate the delivery algorithms proposed in Sections~\eqref{subsec:minG}--\eqref{subsec: phantom}, applied to a specific network setup with \(K = 10\) users divided into (\(J = 2\)) equal-sized groups with spatial multiplexing gains of \(G_{(1)} = 2\) and \(G_{(2)} = 4\). 
%Each group contains \(K_{(1)} = 5\) and \(K_{(2)} = 5\) users, respectively. 
The cache ratio of all users is \(\gamma = 0.2\), and the transmitter is equipped with \(L = 4\) antennas. For simplicity, we assume that users 1, 2, 3, $\cdots$ request files \(A, B, C, \dots\), and for each scheme, we illustrate one sample transmission vector.
% %in~\cite[Theorem 1]{naseritehrani2024cacheaidedmimocommunicationsdof}

%In this example, we revisit and examine the delivery algorithms proposed in Sections~\eqref{sec: Proposed schemes} and~\eqref{sec: phantom}, applied to a specific network setup. The configuration includes \(K = 10\) users, divided into two user sets (\(J = 2\)), with spatial multiplexing gains \(G_{(1)} = 2\) and \(G_{(2)} = 4\), and corresponding user counts \(K_{(1)} = 10\) and \(K_{(2)} = 10\). The cache ratio is \(\gamma = 0.1\), and the transmitter is equipped with \(L = 4\) antennas. For this setup, assuming that each user requests files \(A, B, C, \dots\), while we illustrate only a single sample transmission vector per scheme.
%In this example, the delivery algorithm in~\cite[Theorem.1]{naseritehrani2024cacheaidedmimocommunicationsdof} is reviewed in a particular network setup  
% with $K=10$ users, two user sets $J=2$, $G_{(1)}=2$, $G_{(2)}=4$, $K_{(1)}=10$, $K_{(2)}=10$, $\gamma=1$, and $L = 4$. For this setup, consider the proposed schemes in Sections~\eqref{sec: Proposed schemes},~\eqref{sec: phantom} and assuming each user requests files $A,B,C,\cdots$, while we illustrate only a single sample transmission vector per scheme.

%\subsubsection{The {min}-$G$ scheme}
\subsubsection{\texorpdfstring{The min-$G$ scheme}{The min-G scheme}}

First, we substitute \(\check{G} = 2\) and \(K\gamma = 2\) into~\eqref{eq:optimal_omega_beta_minG} to determine the design parameters \(\Omega_{\check{G}}^* = 4\) and \(\beta_{\check{G}}^* = 2\), and subsequently \(\check{\vartheta} = \binom{10}{2}\), \(\check{\phi} = 14\), and \(\check{S} = 3\binom{10}{4}\).
%Next, in the cache-placement phase, calculate \(\check{\vartheta} = \binom{10}{2}\) and \(\check{\phi} = 14\), resulting in \(\check{S} = 3\binom{10}{4}\). 
As a result, 
for each transmission \(s \in [\check{S}]\), the stream size is \(f(s) = \nicefrac{F}{\check{\vartheta} \check{\phi}}\), \(\Omega(s) = \Omega^*_{\check{G}} = 4\) are served, and \(\beta(s) = \beta^*_{\check{G}} = 2\) streams are delivered to each user. Substituting these values into~\eqref{eq:asym_DoF} and simplifying, the Degrees of Freedom (DoF) of the min-\(G\) scheme are given by:  $\mathrm{DoF}^{*}_{\check{G}} = \Omega^*_{\check{G}} \cdot \beta^*_{\check{G}}= 8$.

Let us consider a transmission $s$ to the user subset \(\CK(s) = \{1, 2, 3, 6\}\), requesting files \(\{A, B, C, F\}\), respectively. To construct transmission vectors for this target user set, the packet index sets are defined as \(\SfP_k = \{\CP \subseteq \CK(s) \backslash \{k\}, |\CP| = 2\}\)~\cite[Theorem 1]{naseritehrani2024cacheaidedmimocommunicationsdof}. This results in: 
\begin{equation}
\begin{aligned}
&\SfP_1 = \{23, 26, 36\}, \quad \SfP_2 = \{13, 36, 16\}, 
\\& \SfP_3 = \{12, 16, 26\}, \quad \SfP_6 = \{12, 13, 23\},
\end{aligned}
\end{equation}
where braces for packet indices are omitted here for notational simplicity but are included in the transmission vectors. Now, for each \(k \in \CK(s)\) and \(\CP \in \SfP_k\), we define \(\CN_{\CP,k} = \{W_{\CP,k}^1, W_{\CP,k}^2\}\), and select \(\beta_{\check{G}}^* = 2 \leq  \check{G}\) subpackets per user \(k\) for the transmission. Without loss of generality, assume the following subpackets are chosen for the first transmission:  
\begin{equation*}
\begin{aligned}
A_{\{23\}}^1, \, A_{\{26\}}^1, \, B_{\{13\}}^1, \, B_{\{16\}}^1, \, C_{\{12\}}^1, \, C_{\{16\}}^1, \, F_{\{12\}}^1, \, F_{\{23\}}^1,
\end{aligned}
\end{equation*}
and the remaining subpackets are reserved for subsequent transmissions. The transmission vector for these subpackets is built as: 
\begin{equation}
\begin{aligned}
\Bx(s) &=  \Bw_{\{23\},1}^1 A_{\{23\}}^1 + \Bw_{\{26\},1}^1 A_{\{26\}}^1 \\
&+ \Bw_{\{13\},2}^1 B_{\{13\}}^1 +  \Bw_{\{16\},2}^1 B_{\{16\}}^1  \\
&+  \Bw_{\{12\},3}^1 C_{\{12\}}^1+  \Bw_{\{16\},3}^1 C_{\{16\}}^1 \\
&+   \Bw_{\{12\},6}^1 F_{\{12\}}^1 +  \Bw_{\{23\},6}^1 F_{\{23\}}^1,\\[3mm]
\end{aligned}
\end{equation}
where based on the delivery algorithm proposed in~\cite[Theorem 1]{naseritehrani2024cacheaidedmimocommunicationsdof}, each user can decode its desired streams without interference. For example, user~1 can decode both $A_{\{23\}}^1$ and $A_{\{26\}}^1$ as the interference from $F_{\{23\}}^1$ is suppressed by beamforming, and the interference from other terms is removed by the cache contents of user~1.

\subsubsection{The grouping scheme}
First, we classify users into \(J=2\) groups based on their number of antennas. As the number of users in both groups is the same, the coded caching gain in each group will be \(K_{(j)}\gamma=1\). Then, using \(G_{(1)}=2\) and \(G_{(2)}=4\) in~\eqref{eq:optimal_omega_beta_grouping}, we can determine group-specific design parameters \((\Omega_{(1)}^*=3, \beta_{(1)}^*=2)\) and \((\Omega_{(2)}^*=2, \beta_{(2)}^*=4)\), and subsequently, we get \({\vartheta}_{(1)} = {\vartheta}_{(2)} = 5\), \({\phi}_{(1)} = 6\), \({\phi}_{(2)} = 4\), \(\check{S}_{(1)} = 20\), and \(\check{S}_{(2)} = 10\).
Substituting these values into~\eqref{eq:total_DoF_grouping} and simplifying, the DoF of the grouping scheme for this network is: $\mathrm{DoF}_{\CJ}^* = 6.86$.

Let us review a sample transmission vector for each group. An arbitrary transmission \(s_1\) for group~1 delivers $\beta(s_1) = \beta^*_{(1)} = 2$ parallel streams, each with size $f(s_1) = \nicefrac{F}{\displaystyle \vartheta_{(1)} \phi_{(1)}}$ = \nicefrac{F}{30}, to a subset $\CK(s_1)$ of users in group~1 with size $\Omega_{(1)}^*=3$. Let us assume \(\CK(s_1) = \{1,2,3\}\), and denote the files requested by users \(1\)--\(3\) as \(A\), \(B\), and \(C\), respectively. The first step in constructing transmission vectors for these users is to define packet index sets \(\SfP_k = \{\CP \subseteq \CK(s_1) \setminus \{k\}, |\CP| = 2\}\) as follows:
\begin{equation}
\begin{aligned}
&\SfP_1 = \{2,3\}, \quad \SfP_2 = \{1,3\}, \quad \SfP_3 = \{1,2\}.
\end{aligned}
\end{equation}
Next, for each \(k \in \CK(s_1)\) and \(\CP \in \SfP_k\), we define \(\CN_{\CP,k} = \{W_{\CP,k}^1, W_{\CP,k}^2\}\), and select \(\beta^*_{(1)} = 2\) subpackets for each  target user. Without loss of generality, assume the following subpackets are chosen:
\begin{equation*}
\begin{aligned}
A_{\{2\}}^1, A_{\{3\}}^1, B_{\{1\}}^1, B_{\{3\}}^1, C_{\{1\}}^1, C_{\{2\}}^1. 
\end{aligned}
\end{equation*}
This results in the transmission vector
\begin{equation}
\begin{aligned}
\Bx(s_1) &=  \Bw_{\{2\},1}^1 A_{\{2\}}^1 + \Bw_{\{3\},1}^1 A_{\{3\}}^1 \\
&+ \Bw_{\{1\},2}^1 B_{\{1\}}^1 +  \Bw_{\{3\},2}^1 B_{\{3\}}^1  \\
&+  \Bw_{\{1\},3}^1 C_{\{1\}}^1+  \Bw_{\{2\},3}^1 C_{\{2\}}^1,
\end{aligned}\\
\end{equation}
which delivers two subpackets to each user in $\CK(s_1) = \{1,2,3\}$ interference-free.

Now, for users in group~2, an arbitrary transmission \(s_2\) delivers $\beta(s_2) = \beta^*_{(2)} = 4$ parallel streams, each with size $f(s_2) = \nicefrac{F}{\displaystyle \vartheta_{(2)} \phi_{(2)}}$ = \nicefrac{F}{20}, to a subset $\CK(s_2)$ of users in group~2 with size $\Omega_{(2)}^*=2$. Let us assume \(\CK(s_2) = \{6,7\}\), and denote the files requested by users \(6\),\(7\) as \(F\) and \(G\), respectively. The packet index sets \(\SfP_k\) are built as
\begin{equation}
\begin{aligned}
\SfP_6 = \{7\}, \quad \SfP_7 = \{6\},
\end{aligned}
\end{equation}
and, assuming $F_{\{7\}}^q, G_{\{6\}}^q$, for all $q \in [4]$ are chosen for transmission, the resulting transmission vector is:
\begin{equation}
\begin{aligned}
\Bx(s_2) &=  \Bw_{\{7\},6}^1 F_{\{7\}}^1 +  \Bw_{\{7\},6}^2 F_{\{7\}}^2 +\Bw_{\{7\},6}^3 F_{\{7\}}^3 \\ &+  \Bw_{\{7\},6}^4 F_{\{7\}}^4 +\Bw_{\{6\},7}^1 G_{\{6\}}^1 +  \Bw_{\{6\},7}^2 G_{\{6\}}^2 \\
&  +\Bw_{\{6\},7}^3 G_{\{6\}}^3 + \Bw_{\{6\},7}^4 G_{\{6\}}^4,
\end{aligned}\\
\end{equation}
which delivers four parallel streams to users~6 and~7 interference-free.

\subsubsection{The phantom scheme}

By definition, $\hat{G}$ can vary between $G_{(1)}=2$ and $G_{(2)} = 4$. Let us assume $\hat{G}=4$. We can select any pair of design parameters \((\hat{\Omega}, \hat{\beta})\) satisfying the linear decodability condition in~\eqref{DoF_bndd}. Let us assume \(\hat{\Omega}=3\) and \(\hat{\beta}=4\). This results in $\hat{\vartheta} =\binom{10}{2}$, $\hat{\phi} = 4$, $S_{\mathrm{MC}} = 120$, $Z_{\mathrm{MC}}=1080$, and $Z_{\mathrm{UC}}=360$. If we can serve 4 parallel streams (i.e., if $d(s) = 4$) with each unicast transmissions, then $S_{\mathrm{UC}} = \frac{Z_{\mathrm{UC}}}{4} = 90$, and using~\eqref{Phantom_delivery61}, we can calculate $\mathrm{DoF}_{\hat{G}} = 6.86$. In the following, we represent an arbitrary multicasting and unicasting transmission indices by $s_{\mathrm{MC}}\in [S_{\mathrm{MC}}] $ and $s_{\mathrm{UC}}\in [S_{\mathrm{UC}}]$, respectively.

\textbf{Multicast transmission design.}
Let us review how the transmission vector is built for the user set \(\CK_{\mathrm{MC}}(s_{\mathrm{MC}}) = \{1, 2, 6\}\). Assume these users request files $A$, $B$, and $F$, respectively. The first step is to build a transmission vector $\hat{\Bx}_{\mathrm{MC}} (s_{\mathrm{MC}})$ for the same set of users, assuming users~1 and~2 also have $\hat{G} = 4$ antennas. To do so, we determine the packet index sets
$\SfP_k = \{\CP \subseteq  \CK_{\mathrm{MC}}(s_{\mathrm{MC}}) \backslash \{k\}, |\CP| = 2\}$ as:
% \begin{equation*}
% \begin{aligned}
% \SfP_1 = \{23, 26, 36\}, \quad \SfP_2 = \{13, 16, 36\}, \quad \SfP_6 = \{12, 13, 23\}.
% \end{aligned}
% \end{equation*}
\begin{equation*}
\begin{aligned}
\SfP_1 = \{26\}, \quad \SfP_2 = \{16\}, \quad \SfP_6 = \{12\}.
\end{aligned}
\end{equation*}
Then, 
for each combination \(\CP \in \SfP_k\), we define \(\CN_{\CP, k} = \{W_{\CP, k}^1, W_{\CP, k}^2, W_{\CP, k}^3, W_{\CP, k}^4\}\), and select $\hat{\beta} = 4$ subpackets for each user $k \in  \CK_{\mathrm{MC}}(s_{\mathrm{MC}})$ to be transmitted by $\hat{\Bx}_{\mathrm{MC}} (s_{\mathrm{MC}})$. Assuming we have selected
\begin{equation*}
\begin{aligned}
 A_{\{26\}}^1, A_{\{26\}}^2, &A_{\{26\}}^3, A_{\{26\}}^4, B_{\{16\}}^1, B_{\{16\}}^2, B_{\{16\}}^3, B_{\{16\}}^4,\\&
F_{\{12\}}^1, F_{\{12\}}^2, F_{\{12\}}^3, F_{\{12\}}^4,
%A_{\{23\}}^1, A_{\{26\}}^1, &A_{\{36\}}^1, A_{\{23\}}^2, B_{\{13\}}^1, B_{\{16\}}^1, B_{\{36\}}^1, B_{\{13\}}^2,\\&
%F_{\{12\}}^1, F_{\{13\}}^1, F_{\{23\}}^1, F_{\{12\}}^2,
\end{aligned}
\end{equation*}
the resulting multicasting transmission vector is:
\begin{equation}
\begin{aligned}
\hat{\Bx}_{\mathrm{MC}}(s_{\mathrm{MC}}) &=  \Bw_{\{26\},1}^1 A_{\{26\}}^1 + \Bw_{\{26\},1}^2 A_{\{26\}}^2 \\
&+ \Bw_{\{26\},1}^3 A_{\{26\}}^3 + \Bw_{\{26\},1}^4 A_{\{26\}}^4  \\
&+ \Bw_{\{16\},2}^1 B_{\{16\}}^1 +  \Bw_{\{16\},2}^2 B_{\{16\}}^2  \\
&+  \Bw_{\{16\},2}^3 B_{\{16\}}^3 +  \Bw_{\{16\},2}^4 B_{\{16\}}^4 \\
&+   \Bw_{\{12\},6}^1 F_{\{12\}}^1 +  \Bw_{\{12\},6}^2 F_{\{12\}}^2 
\\
&+ \Bw_{\{12\},6}^3 F_{\{12\}}^3 +  \Bw_{\{12\},6}^4 F_{\{12\}}^4.\\
\end{aligned}
\end{equation}
% \begin{equation}
% \begin{aligned}
% \hat{\Bx}_{\mathrm{MC}}(s_{\mathrm{MC}}) &=  \Bw_{\{23\},1}^1 A_{\{23\}}^1 + \Bw_{\{26\},1}^1 A_{\{26\}}^1 \\
% &+ \Bw_{\{36\},1}^1 A_{\{36\}}^1 + \Bw_{\{23\},1}^2 A_{\{23\}}^2  \\
% &+ \Bw_{\{13\},2}^1 B_{\{13\}}^1 +  \Bw_{\{16\},2}^1 B_{\{16\}}^1  \\
% &+  \Bw_{\{36\},2}^1 B_{\{36\}}^1 +  \Bw_{\{13\},2}^2 B_{\{13\}}^2 \\
% &+   \Bw_{\{12\},6}^1 F_{\{12\}}^1 +  \Bw_{\{23\},6}^1 F_{\{23\}}^1 
% \\
% &+ \Bw_{\{13\},6}^1 F_{\{13\}}^1 +  \Bw_{\{23\},6}^2 F_{\{23\}}^2.\\
% \end{aligned}
% \end{equation}
%then after removing $\hat{\beta}-\beta_k$ subpackets from each user $k$ (outlined in the Section~\ref{subsec: phantom}, two packets removal from the packets of users k=1,2,3), in $\hat{\Bx}_{\mathrm{MC}}(1)$, we obtain ${\Bx}_{\mathrm{MC}}(1)$:
%
Finally, by removing \(\hat{\beta} - \beta_k\) subpackets for each user \(k\) (linear decodability is maintained by eliminating \(\hat{\beta} - \beta_k\) subpackets for each user \(k\), as outlined in Theorem~\ref{Th:Linear_Dec})  the final multicasting transmission vector ${\Bx}_{\mathrm{MC}}(s_{\mathrm{MC}})$ is given as:
%
% \begin{equation}\label{eq:x_mc_exmp}
% \begin{aligned}
% {\Bx}_{\mathrm{MC}}(s_{\mathrm{MC}}) &=  \Bw_{\{23\},1}^1 A_{\{23\}}^1 + \Bw_{\{26\},1}^1 A_{\{26\}}^1  \\
% &+ \Bw_{\{13\},2}^1 B_{\{13\}}^1 +  \Bw_{\{16\},2}^1 B_{\{16\}}^1  \\
% &+   \Bw_{\{12\},6}^1 F_{\{12\}}^1 +  \Bw_{\{23\},6}^1 F_{\{23\}}^1 
% \\
% &+ \Bw_{\{13\},6}^1 F_{\{13\}}^1 +  \Bw_{\{23\},6}^2 F_{\{23\}}^2,\\
% \end{aligned}
% \end{equation}
\begin{equation}\label{eq:x_mc_exmp}
\begin{aligned}
{\Bx}_{\mathrm{MC}}(s_{\mathrm{MC}}) &=  \Bw_{\{26\},1}^1 A_{\{26\}}^1 + \Bw_{\{26\},1}^4 A_{\{26\}}^4  \\
&+ \Bw_{\{16\},2}^1 B_{\{16\}}^1 +  \Bw_{\{16\},2}^3 B_{\{16\}}^3  \\
&+   \Bw_{\{12\},6}^1 F_{\{12\}}^1 +  \Bw_{\{12\},6}^2 F_{\{23\}}^2 
\\
&+ \Bw_{\{12\},6}^3 F_{\{12\}}^3 +  \Bw_{\{12\},6}^4 F_{\{12\}}^4,\\
\end{aligned}
\end{equation}
where $A_{\{26\}}^1$, $A_{\{26\}}^4$, $B_{\{16\}}^1$, and $B_{\{16\}}^3$ are left to be delivered in subsequent unicast transmissions. It can be easily verified that using ${\Bx}_{\mathrm{MC}}(s_{\mathrm{MC}})$ in~\eqref{eq:x_mc_exmp}, users~1, 2, and~6 can decode 2, 2, and 4 parallel streams, respectively.
% The linear decodability is maintained by eliminating \(\hat{\beta} - \beta_k\) subpackets for each user \(k\), as outlined in Theorem~\ref{Th:Linear_Dec}. This ensures that each user can decode its missing subpackets simultaneously and without interference.\\

\textbf{Unicast transmission design.}
All the subpackets removed from multicast transmissions are later sent by unciasting. For example, let us consider how $A_{\{26\}}^2$, $A_{\{26\}}^3$, $B_{\{16\}}^2$, and $B_{\{16\}}^4$ could be transmitted to users~1 and~2 with a single unicast transmission. This can be achieved simply by transmitting:
\begin{equation}
\begin{aligned}
{\Bx}_{\mathrm{UC}}(s_{\mathrm{UC}}) &=  \Bw_{\{26\},1}^2 A_{\{26\}}^2 + \Bw_{\{26\},1}^3 A_{\{26\}}^3  \\[2mm]
+&\Bw_{\{16\},2}^2 B_{\{16\}}^2 +  \Bw_{\{16\},2}^4 B_{\{16\}}^4.
\end{aligned}
\end{equation}

\end{exmp}

%\subsection{Proof of DoF Comparisons in Lemma~\ref{lm: DoF_lemma_grouping_min_G}}
% \subsection{DoF Comparison for $J=2$}
% \label{app_B}
% \input{Proof_DoF_Comparisons}

\end{document}